\documentclass[runningheads]{llncs}

\usepackage{scrextend}
\usepackage{amsmath}
\usepackage{amssymb,enumerate,url}
\usepackage{tikz}
\usetikzlibrary{backgrounds,calc,positioning}
\usepackage[linguistics]{forest}
\usetikzlibrary{arrows.meta}
\usepackage{tabularx}
\newcolumntype{L}[1]{>{\raggedright\arraybackslash}p{#1}} 
\newcolumntype{C}[1]{>{\centering\arraybackslash}p{#1}} 
\newcolumntype{R}[1]{>{\raggedleft\arraybackslash}p{#1}} 
\usepackage{booktabs,multirow} 
\usepackage{hyperref}

\usetikzlibrary{calc,arrows,decorations.pathreplacing}\usetikzlibrary{arrows,automata,positioning}
\tikzstyle{emptyKnot}=[scale=0.8]
\tikzstyle{knot}=[circle, fill,scale=0.4]
\tikzstyle{smallKnot}=[circle, fill,scale=0.3]

\allowdisplaybreaks

\newcommand{\ut}{\mathcal{T}}
\newcommand{\bt}{\mathcal{B}}

\newcommand{\N}{\mathbb{N}}

\newcommand{\lh}{h^\ell}

\newcommand{\parent}{\operatorname{parent}}

\newcommand{\fname}[1]{#1}
\newcommand{\ffcns}{\text{fcns}}

\newcommand{\tildem}{n}

\title{A Comparison of Empirical Tree Entropies}

\author{
  Danny Hucke\and
  Markus Lohrey\and 
  Louisa Seelbach Benkner
}
\institute{
University of Siegen, Germany,
\email{\{hucke,lohrey,seelbach\}@eti.uni-siegen.de}}

\begin{document}

\maketitle

\begin{abstract}
Whereas for strings, higher-order empirical entropy is the standard entropy measure,
several different notions of empirical entropy for trees have been proposed in the past, notably label entropy,
degree entropy, conditional versions of the latter two, and empirical entropy of trees (here, called label-shape entropy).
In this paper, we carry out a systematic comparison of these entropy measures. We underpin our theoretical
investigations by experimental results with real XML data.
\end{abstract}

\section{Introduction}

In the area of string compression the notion of higher order empirical entropy yields a well established 
measure for the compressibility of a string. Roughly speaking, the $k^{th}$-order empirical entropy of a string $s$
is our expected uncertainty about the symbol at a certain position, given the $k$-preceding symbols.
In fact, except for some modifications (as the $k^{th}$-order modified empirical entropy from \cite{Manzini01}) the authors are not aware of any other empirical entropy measure for strings (``empirical'' refers to the fact
that the entropy is defined for the string itself and not a certain probability distribution on strings). For many string
compressors, worst-case bounds on the length of a compressed string $s$ in terms of the $k^{th}$-order empirical
entropy are known \cite{Gan18,Manzini01,NaOch18}. For further aspects of higher-order empirical entropy see \cite{Gagie06a}.

If one goes from strings to trees the situation becomes different. Let us first mention that the area of tree compression
(and compression of structured data in general) is currently a very active area, which is motivated by the appearance
of large tree data in applications like XML processing, see e.g. \cite{FerraginaLMM05,FerraginaLMM09,HuckeL19,Ganczorz20,HuckeLS19,JanssonSS12,LohreyMM13,Prezza20}.
In recent years, several notions of empirical tree entropy have been proposed
with the aim of quantifying the compressibility of a given tree.
Let us briefly discuss these entropies in the following paragraphs (all entropies below are
unnormalized; the corresponding normalized entropies are obtained by dividing by the tree size).

 Ferragina et al.~\cite{FerraginaLMM05,FerraginaLMM09} introduced the $k^{th}$-order label entropy $H^\ell_k(t)$ of a node-labeled unranked\footnote{Unranked means that there is no bound on the number of children. Moreover, we only consider ordered trees, where the children of a node are linearly ordered.} tree $t$.
Its normalized version is the expected uncertainty about the label of a node $v$, given the so-called {\em $k$-label-history} of $v$ which 
consists of the $k$ first labels on the unique path from 
$v$'s parent node to the root.
Note that the $k^{th}$-order label entropy  is not useful for unlabeled trees since it is independent of the tree shape.

In \cite{JanssonSS12}, Jansson et al. introduce the {\em degree entropy} $H^{\deg}(t)$, which is the (unnormalized) $0^{th}$-order empirical
entropy of the node degrees occurring in the unranked tree $t$.
The degree entropy is mainly made for unlabeled trees since it ignores node labels.
But in combination with label entropy it yields a reasonable measure for the compressibility of a tree: 
every node-labeled unranked tree of size $n$ in which $\sigma$ many different node labels occur
can be stored in $H^\ell_k(t) + H^{\deg}(t) + o(n \log \sigma)$ bits assuming that $\sigma$ is not too big; 
see Theorem~\ref{theorem-ganczorzentropybounds}.\footnote{Formally, we should always replace $\sigma$ by $\max\{2,\sigma\}$ in order
to avoid the pathological case that the term $o(n \log \sigma)$ vanishes. The same holds for terms $\log_\sigma(n)$ that will occur later.}
Note that the (unnormalized) degree entropy of a binary tree with $n$ leaves converges to $2n - o(n)$ since
a binary tree with $n$ leaves has exactly $n-1$ nodes of degree $2$. 

Recently, Ganczorz \cite{Ganczorz20} defined relativized versions of $k^{th}$-order label entropy and degree entropy:
The $k^{th}$-order degree-label entropy $H^{\deg,\ell}_k(t)$ and the $k^{th}$-order label-degree entropy $H^{\ell,\deg}_k(t)$.
The normalized version of $H^{\deg,\ell}_k(t)$ is the expected uncertainty about the label of a node $v$ of $t$, given (i) the 
$k$-label-history of $v$ and (ii) the degree of $v$, whereas the normalized version of $H^{\ell,\deg}_k(t)$ is
the expected uncertainty about the degree of a node $v$, given (i) the 
$k$-label-history of $v$ and (ii) the label of $v$.  Ganczorz \cite{Ganczorz20} proved that every node-labeled unranked tree of size $n$ can be stored in 
$H^\ell_k(t) + H^{\ell,\deg}_k(t) + o(n \log \sigma)$ bits as well as in $H^{\deg}(t) + H_k^{\deg,\ell}(t) + o(n \log \sigma)$ bits 
(again assuming $\sigma$ is not too big), see Theorem~\ref{theorem-ganczorzentropybounds}. Note that for unlabeled trees $t$, we have $H^\ell_k(t) + H^{\ell,\deg}_k(t)=H^{\deg}(t) + H_k^{\deg,\ell}(t)=H^{\deg}(t)$, which for binary trees equals the information theoretic upper bound $2n - o(n)$ (with $n$ the number of leaves).

Motivated by the inability of the existing entropies for measuring the compressibility of unlabeled binary trees, we 
introduced in \cite{HuckeLS19} a new entropy for binary trees (possibly with labels) that we called $k^{th}$-order empirical
entropy $H_k(t)$. In order to distinguish it better from the existing tree entropies we prefer the term {\em $k^{th}$-order label-shape entropy}
in this paper. The main idea is to extend $k$-label-histories in a binary tree by adding to the labels of the $k$ predecessors of a node $v$ also
the $k$ last directions ($0$ for left, $1$ for right) on the path from the root to $v$. We call this extended label history simply the 
$k$-history of $v$. The normalized version of $H_k(t)$ is the expected uncertainty about the pair consisting of the label of a node and 
the information whether it is a leaf or an internal node, given the $k$-history of the node. The main result of \cite{HuckeLS19} 
states that a node-labeled binary tree $t$ can be stored in $H_k(t) + o(n \log \sigma)$ bits using a grammar-based code based on so-called
tree straight-line programs. We also defined in  \cite{HuckeLS19}  the $k^{th}$-order label-shape entropy of an unranked node-labeled tree $t$
by taking the $k^{th}$-order label-shape entropy of the first-child next-sibling encoding of $t$.

\definecolor{dgreen}{rgb}{0,0.6,0}
\newcommand{\dgreen}{\color{dgreen}}
\begin{figure*}[t]
		\centering
		\tikzstyle{lts} = [->, >=stealth]
		\tikzstyle{state} = [inner sep = .7mm]
		\scalebox{1}{
			\begin{tikzpicture}[lts]
			
			\node [state] (1) {$H_k$};
			\node [state, above left = 3cm and 2cm of 1] (2) {$H^{\deg} + H_k^{\deg,\ell}$};
			\node [state, above right = 3cm and 2cm of 1] (3) {$H^\ell_k + H_k^{\ell,\deg}$};
			
			\draw[dgreen] (2) to[bend left=20] node[above=-.5mm]{$\forall \geq$} (3);
			\draw[red] (3) to[bend left=20] node[above=-.5mm]{$\exists \, o$} (2);
			
			\draw[red] (1) to[bend right=20] node[right=-.3mm,pos=.5]{$\exists \, o$} (2);
			\draw[red] (2) to[bend right=20] node[left=-.4mm]{$\exists \, o$} (1);
			
			\draw[red] (3) to[bend left=20] node[right=0mm]{$\exists \, o$} (1);
			\draw[red] (1) to[bend left=20] node[left=0mm]{$\exists \, o$} (3);
			\node [state, left = 1.7cm of 2] (4) {$H^{\deg} + H_k^{\ell}$};
			\draw[-,blue] (2) to node[above=0mm]{$\forall \Theta$} (4);
			\end{tikzpicture}
		}
		\caption{Comparison of the entropy notions for unranked node-labeled trees. The meaning of the red and green arrows is explained in the main text.}
		\label{fig-comparison}
	\end{figure*}
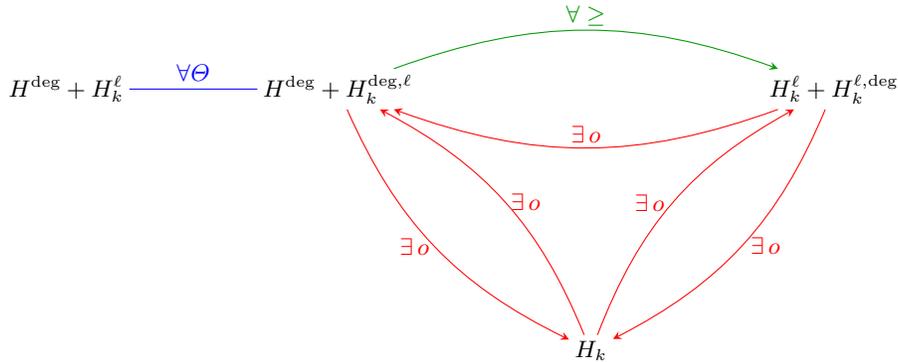

The goal of this paper is to compare the entropy variants  $H^\ell_k(t) + H^{\deg}(t)$, $H^\ell_k(t) + H^{\ell,\deg}_k(t)$, $H^{\deg}(t) + H_k^{\deg,\ell}(t)$, and $H_k(t)$.
Our results for unranked node-labeled trees are summarized in Figure~\ref{fig-comparison}.
Let us explain the meaning of the arrows in Figure~\ref{fig-comparison}: For two entropy notions $H$ and $H'$, a red arrow 
$$H \;\textcolor{red}{\xrightarrow{\;\exists \, o\;}} \;H'$$
means that there is a sequence of unranked node-labeled trees $t_n$ ($n \geq 1$) such that 
(i) the function $n \mapsto |t_n|$ is strictly increasing and (ii) $H(t_n) \leq o(H'(t_n))$
(in most cases we prove an exponential separation).
The meaning of the green arrow is that $H^{\deg}(t) + H_k^{\deg,\ell}(t) \geq  H^\ell_k(t) + H_k^{\ell,\deg}(t)$ for every unranked node-labeled tree $t$, whereas the blue edge means that 
$H^{\deg}(t) + H_k^{\deg,\ell}(t)$  and $H^{\deg}(t) + H^\ell_k(t)$ are equivalent up to fixed multiplicative constants (which are 1 and 2).

We also investigate the relationship between the entropies for node-labeled binary trees and unranked unlabeled trees
(the case of unlabeled binary trees is not really interesting as explained above). An unranked unlabeled tree $t$ of size $n$ can be represented with $H^{\deg}(t) + o(n)$ 
bits \cite{JanssonSS12}. Here, we prove that $H_k(t) \leq 2H^{\deg}(t) + 2\log_2(n)+4$.

Finally, we underpin our theoretical investigations by experimental results with real XML data from XMLCompBench\footnote{\url{http://xmlcompbench.sourceforge.net}}. For each XML we consider the corresponding tree structure $t$ (obtained by removing all text values and attributes) and compute
$H^\ell_k(t) + H^{\deg}(t)$, $H^\ell_k(t) + H^{\ell,\deg}_k(t)$, $H^{\deg}(t) + H_k^{\deg,\ell}(t)$, and $H_k(t)$.
The results are summarized in Table~\ref{table1}. 
Our experiments indicate that the upper bound on the bits needed by the compressed data structure in~\cite{HuckeLS19} is the strongest for real XML data
since the $k^{th}$-order label-shape entropy (for $k>0$) is significantly smaller than all other values for all XMLs that we have examined.

Let us remark that Ganczorz's succinct tree representations \cite{Ganczorz20} that achieve (up to low-order terms) the entropies 
$H^\ell_k(t) + H^{\ell,\deg}_k(t)$ and $H^{\deg}(t) + H_k^{\deg,\ell}(t)$, respectively, allow constant query times for a large number
of tree queries. For the entropy $H_k(t)$ such a result is not known. The tree representation from~\cite{HuckeLS19} is based
on tree straight-line programs, which can be queried in logarithmic time (if we assume logarithmic height of the grammar, which
can be enforced by \cite{GJL19}).

\section{Preliminaries}\label{sec-prelim}
In this section, we introduce some basic definitions. With $\N$ we denote the natural numbers including $0$.
Let $w = a_1 a_2 \cdots a_l \in \Gamma^*$ be a word over an alphabet $\Gamma$. With $|w|=l$ we denote the 
length of $w$. The empty word is denoted
by $\varepsilon$. 
We use the standard $\mathcal{O}$-notation. If $b>0$ is a constant, then
we just write $\mathcal{O}(\log n)$ for $\mathcal{O}(\log_b n)$. 
We make the convention that $0 \cdot \log(0) = 0$ and $0 \cdot \log(x/0)=0$ for $x \geq 0$.
We need the well-known log-sum inequality, see e.g. \cite[Theorem~2.7.1]{CoTh06}:

\begin{lemma}[Log-Sum inequality]\label{logsum}
Let $a_1, a_2, \dots, a_l,b_1, b_2, \dots, b_l \geq 0$ be real numbers. Moreover, let $a = \sum_{i=1}^l a_i$ and $b=\sum_{i=1}^l b_i$. Then
\begin{align*}
 a \log_2 \left(\frac{b}{a}\right) \geq \sum_{i=1}^l a_i \log_2\left(\frac{b_i}{a_i}\right).
\end{align*}
\end{lemma}

\subsection{Unranked trees} \label{sec-unrankedtrees}
Let $\Sigma$ denote a finite alphabet of size $|\Sigma|=\sigma$. Later, we will need a fixed, distinguished symbol from $\Sigma$ that we will denote with $\Box \in \Sigma$. Throughout the paper, we consider $\Sigma$-\emph{labeled unranked ordered trees}, where ``$\Sigma$-labeled'' means that every node is labeled by a character from the alphabet $\Sigma$, ``ordered'' means that the children of a node are totally ordered, and ``unranked'' means that the number of  children of a node (also called its \emph{degree})
can be any natural number. In particular, the degree of a node does not depend on the node's label or vice versa. Let us denote by $\ut(\Sigma)$ the set of all such trees. 
Formally, the set $\ut(\Sigma)$ is inductively defined as the smallest set of expressions such that
 if $a \in \Sigma$ and $t_1, \ldots, t_n \in \ut(\Sigma)$ then also $a(t_1\cdots t_n) \in \ut(\Sigma)$. This expression represents a tree with root $a$ whose direct subtrees are $t_1, \ldots, t_n$. Note that for the case $n=0$ we obtain the tree $a()$, for which we also write $a$.
The \emph{size} $|t|$ of $t \in \ut(\Sigma)$ is the number of occurrences of labels from $\Sigma$ in $t$,
i.e., $a(t_1\cdots t_n) = 1+\sum_{i=1}^n |t_i|$. We will identify an unranked tree as a graph with nodes and edges in the usual way, where each node is labeled with a symbol from $\Sigma$.
Let $V(t)$ denote the set of nodes of a tree $t \in \ut(\Sigma)$. We have $|V(t)| = |t|$.
The label of a node $v \in V(t)$  is denoted with $\ell(v) \in \Sigma$.
Moreover, we write $\deg(v) \in \mathbb{N}$ for the degree of $v$ (its number of children).
An important special case of unranked trees are \emph{unlabeled unranked trees}: They can be considered as labeled unranked trees over a unary alphabet (e.g. $\Sigma=\{a\}$).  

 For a node $v \in V(t)$ of a tree $t$, we define its \emph{label-history} $\lh(v) \in \Sigma^*$ inductively as follows:  
For the root node $v_0$, we set $\lh(v_0)=\varepsilon$ and for a child node $w$ of a node $v$ of $t$, we set $\lh(w)=\lh(v) \, \ell(v)$.
In other words:  
$\lh(v)$ is obtained by concatenating the node labels along the unique path from the root to $v$. Note that the symbol that labels $v$ is not part of the label-history of $v$. The $k$-\emph{label-history} $\lh_k(v)$ of a tree node $v \in V(t)$ is  defined as the length-$k$-suffix of $\Box^k \lh(v)$, where $\Box$ is a fixed dummy symbol in $\Sigma$. This means that if the depth of $v$ in $t$ is greater than $k$, then $\lh_k(v)$ describes the last $k$ node labels along the path from the root to node $v$. If the depth of $v$ in $t$ is at most $v$, then we pad its label-history $\lh(v)$ with the symbol $\Box$ such that $\lh_k(v) \in \Sigma^k$. In general, there are several possibilities how to define the $k$-label-history of nodes of depth smaller than $k$, several alternatives are discussed in \cite{Journalarxivversion}.
 
For $z \in \Sigma^k$, $a \in \Sigma$ and $i \in \mathbb{N}$ we set 
 \begin{eqnarray}
\tildem_z^t &=& |\{v \in V(t) \mid \lh_k(v) = z\}|, \label{n_z^t} \\
\tildem_{z, a}^t &=& |\{v \in V(t) \mid \lh_k(v) = z \text{ and } \ell(v) =a\}|, \label{n_z,a^t} \\
\tildem_{i}^t &=& |\{v \in V(t) \mid \deg(v) =i\}|, \label{n_i^t} \\
\tildem_{z, i}^t &=& |\{v \in V(t) \mid \lh_k(v) = z \text{ and } \deg(v) =i\}|, \label{n_z,i} \\
\tildem_{z, i,a}^t &=& |\{v \in V(t) \mid \lh_k(v) = z, \, \ell(v) = a \text{ and } \deg(v) = i\}|. \label{n_z,i,a} 
\end{eqnarray}
In order to avoid ambiguities in these notations we should assume that $\Sigma \cap \N = \emptyset$. Moreover,
when writing $n^t_{z,i}$ (resp., $n^t_{z,a}$) then, implicitly, $i$ (resp., $a$) always belongs to $\N$ (resp., $\Sigma$).

\subsection{Binary trees} \label{sec-binarytrees}

An important subset of $\ut(\Sigma)$ is the set $\bt(\Sigma)$ of \emph{labeled binary trees} over the alphabet $\Sigma$: A binary tree is a tree in $\ut(\Sigma)$, where every node has either exactly two children or is a leaf. 
Formally, $\bt(\Sigma)$ is inductively defined as the smallest set of terms over $\Sigma$ such that 
\begin{itemize}
\item $\Sigma \subseteq \bt(\Sigma)$ and
\item if $t_1, t_2 \in \bt(\Sigma)$ and $a \in \Sigma$, then $a(t_1, t_2) \in \bt(\Sigma)$.
\end{itemize}
An \emph{unlabeled binary tree} can be considered as a binary tree over the unary alphabet $\Sigma = \{a\}$.
The \emph{first-child next-sibling encoding} (or shortly \emph{fcns-encoding}) transforms an unranked  tree 
$t \in \ut(\Sigma)$ into a binary tree $t \in \bt(\Sigma)$. 
We define it more generally for an ordered sequence of unranked trees $s = t_1 t_2 \cdots t_n$ (a so called forest) inductively as follows 
(recall that $\Box \in \Sigma$ is a fixed distinguished symbol in $\Sigma$):
 \begin{itemize}
\item $\ffcns(s) = \Box$ for $n=0$ and
\item if $n \geq 1$ and $t_1 = a(t'_1 \cdots t'_m)$ then
$\ffcns(s)=a(\ffcns(t'_1 \cdots t'_m), \ffcns(t_2\cdots t_n))$.
\end{itemize}
Thus, the left (resp. right) child of a node in $\ffcns(s)$ is the first child (resp., right sibling) of the node in $s$ or a $\Box$-labeled leaf, if it does not exist. 

For the special case of binary trees, we extend the label history of a node to its full history, which we just call its history. 
 Intuitively, the history of a node $v$ records all information that can be
obtained by walking from the root of the tree straight down to the node $v$. In addition to the node labels this also
includes the directions (left/right) of the decending edges.
Let
\begin{align*}
\mathcal{L} = \left(\Sigma\{0,1\}\right)^* = \{a_1i_1a_2i_2\cdots a_ni_n \mid n \geq 0, a_k \in \Sigma, i_k \in \{0,1\} \text{ for } 1 \leq k \leq n \},
\end{align*}
and for an integer $k \geq 0$ let $\mathcal{L}_k = \{w \in \mathcal{L} \mid |w |= 2k\}$. For a node $v$ of a binary tree $t$, we define its \emph{history} $h(v)$ inductively as follows: For the root node $v_0$, we set $h(v_0) = \varepsilon$. For a left child node $w$ of a node $v$ of $t$, we set $h(w) = h(v) \ell(v) 0$ and for a right child node $w$ of $v$, we set $h(w) = h(v) \ell(v) 1$
(recall that $\ell(v)$ is the label of $v$). That is, in order to obtain $h(v)$, while decending in the tree from the root node to the node $v$, we alternately concatenate symbols from $\Sigma$ with bits from $\{0,1\}$ such that the symbol from $\Sigma$ corresponds to the label of the current node and the bit $0$ (resp., $1$) indicates that we decend to the left (resp., right) child node. Note that the symbol that labels $v$ is not part of the history $h(v)$. The $k$-\emph{history} of a node $v$ is then defined as the $2k$-length suffix of the word $(\Box0)^kh(v)$, where $\Box$ is again a fixed dummy symbol in $\Sigma$.
This means that if the depth of $v$ in $t$ is greater than $k$, then $h_k(v)$ describes the last $k$ directions and node labels along the path from the root to node $v$. If the depth of $v$ in $t$ is at most $k$, then we pad the history of $v$ with $\Box$'s and zeroes such that $h_k(v) \in \mathcal{L}_k$. Again, there are alternative ways how to deal with nodes of depth smaller than $k$, which are discussed in \cite{Journalarxivversion}.

For a node $v$ of a binary tree we define $\lambda(v)=(\ell(v),\deg(v)) \in \Sigma \times \{0,2\}$.
For $z \in \mathcal{L}_k$ and $\tilde{a} \in \Sigma \times \{0,2\}$, we  finally define
\begin{eqnarray}
m_z^t &=& |\{v \in V(t) \mid h_k(v) = z\}|, \label{m_z^t} \\
m_{z, \tilde{a}}^t &=& |\{v \in V(t) \mid h_k(v) = z \text{ and } \lambda(v) = \tilde{a}\}|. \label{m_z,tildea^t}
\end{eqnarray}

\section{Empirical entropy for trees}

In this section we formally define the various entropy measures that were mentioned in the introduction.
Note that in all cases we define so-called unnormalized entropies, which has the advantage that we do not
have to multiply with the size of the tree in bounds for the encoding size of a tree. Note that in \cite{FerraginaLMM05,FerraginaLMM09,Ganczorz20,JanssonSS12}
the authors define normalized entropies. In each case, one obtains the normalized entropy by dividing the corresponding
unnormalized entropy by the tree size.

\subsection{Label entropy}
The first notion of empirical entropy for trees was introduced in \cite{FerraginaLMM05}. In order to distinguish the notions, we will call the empirical entropy from \cite{FerraginaLMM05} \emph{label entropy}. 
It is defined for unranked labeled trees $t \in  \ut(\Sigma)$: The $k^{th}$-order \emph{label entropy} $H^\ell_k(t)$
of $t$ is defined as follows, where $\tildem_z^t$ and $\tildem_{z,a}^t$ are from \eqref{n_z^t} and \eqref{n_z,a^t}, respectively:
\begin{equation} \label{labelentropy}
H_k^{\ell}(t) = \sum_{z \in \Sigma^k}\sum_{a \in \Sigma}\tildem_{z,a}^t\log_2\left(\frac{\tildem_z^t}{\tildem_{z,a}^t}\right).
\end{equation}
We remark that in \cite{FerraginaLMM05}, it is actually not explicitly specified how to deal with nodes, whose label-history is shorter than $k$. 
There are three natural variants: 
\begin{enumerate}[(i)]
\item padding the label-histories with a symbol $\Box \in \Sigma$ (this is our choice), 
\item padding label-histories with a symbol $\diamond \notin \Sigma$, or equivalently, allowing label-histories of length smaller than $k$, and 
\item ignoring nodes whose label-history is shorter than $k$.
\end{enumerate}
However, similar considerations as presented in the appendix of \cite{Journalarxivversion} show that these approaches yield the same $k^{th}$-order label entropy
up to an additional additive term of at most $m^{\scriptscriptstyle{<}} (1+1/\ln(2)+ \log_2(\sigma |t|/m^{\scriptscriptstyle{<}}))$, where $m^{\scriptscriptstyle{<}}$ is the number of nodes at depth less than $k$ in $t$.

Moreover, we remark that in the original paper on label entropy \cite{FerraginaLMM05},
the authors quite often assume disjoint label alphabets for inner nodes and leaves, i.e., inner nodes are labeled with symbols from an alphabet $\Sigma_1$ while
leaves are labeled with symbols from an alphabet $\Sigma_2$ with $\Sigma_1 \cap \Sigma_2 = \emptyset$. 
We will not make this assumption in the following.

\subsection{Degree entropy}

Another notion of empirical entropy for trees is the entropy measure from \cite{JanssonSS12}, which we call \emph{degree entropy}. Degree entropy is primarily made for unlabeled unranked trees, as it ignores node labels. Nevertheless the definition works for trees $t \in \ut(\Sigma)$ over any alphabet $\Sigma$. For a tree $t \in \ut(\Sigma)$, 
the degree entropy $H^{\deg}(t)$ is the $0^{th}$-order entropy of the node degrees ($n^t_i$ is from \eqref{n_i^t}):
\begin{align*}
H^{\deg}(t) = \sum_{i=0}^{|t|} n_i^t \log_2 \left(\frac{|t|}{n_i^t}\right).
\end{align*}
Note that this definition completely ignores node labels.  For the special case of unlabeled trees the following result was shown in 
\cite{JanssonSS12}:

\begin{theorem}[\mbox{\cite[Theorem~1]{JanssonSS12}}]\label{theorem-degentropybound}
Let $t$ be an unlabeled unranked tree. Then $t$ can be represented in 
$
H^{\deg}(t) + \mathcal{O}(|t| \log \log (|t|)/\log |t|)
$
many bits. 
\end{theorem}

\subsection{Label-degree entropy and degree-label entropy}

Recently, two combinations of the label entropy from \cite{FerraginaLMM05} and the degree entropy from \cite{JanssonSS12} were proposed in \cite{Ganczorz20}.
We call these two entropy measures \emph{label-degree entropy}  and \emph{degree-label entropy}.
Both notions are defined for unranked node-labeled trees. Let $t \in \ut(\Sigma)$ be such a tree.
The $k^{th}$-order \emph{label-degree entropy} $H_k^{\ell,\deg}(t)$ of $t$ from \cite{Ganczorz20} is defined as follows, where $\tildem_{z,a}^t$ and $\tildem_{z,i,a}^t$ are from \eqref{n_z,a^t}  and \eqref{n_z,i,a}, respectively:
\begin{align*}
H_k^{\ell,\deg}(t)= \sum_{z \in \Sigma^k}\sum_{a \in \Sigma}\sum_{i=0}^{|t|} \tildem_{z,i,a}^t\log_2\left(\frac{\tildem_{z,a}^t}{\tildem_{z,i,a}^t}\right).
\end{align*}
The $k^{th}$-order \emph{degree-label entropy} $H_k^{\deg,\ell}(t)$ of $t$ from \cite{Ganczorz20} is defined as follows, where
$\tildem_{z,i}^t$ and $\tildem_{z,i,a}^t$ are from \eqref{n_z,i}  and \eqref{n_z,i,a}, respectively:
\begin{align*}
H_k^{\deg,\ell}(t)= \sum_{z \in \Sigma^k}\sum_{i=0}^{|t|}\sum_{a \in \Sigma} \tildem_{z,i,a}^t\log_2\left(\frac{\tildem_{z,i}^t}{\tildem_{z,i,a}^t}\right).
\end{align*}
In order to deal with nodes whose label-history is shorter than $k$ one can again choose one of the three alternatives (i)--(iii) that were mentioned after
\eqref{labelentropy}. In \cite{Ganczorz20}, variant (ii) is chosen, while the above definitions correspond to choice (i). 
However, similar considerations as presented in the appendix of \cite{Journalarxivversion} show again that these approaches are basically equivalent, except for an additional additive term of at most $m^{\scriptscriptstyle{<}} (1/\ln(2)+\log_2(\sigma |t|/m^{\scriptscriptstyle{<}}))$ in the case of the degree-label entropy, respectively, $m^{\scriptscriptstyle{<}} (1/\ln(2)+\log_2|t|)$ in the case of the label-degree entropy, where $m^{\scriptscriptstyle{<}}$ is the number of nodes at depth less than $k$.
In \cite{Ganczorz20}, the following lemma is shown, which relates the degree-label entropy to the label entropy $H_k^{\ell}(t)$ from \eqref{labelentropy}
and the label-degree entropy to the degree entropy:

\begin{lemma} \label{lemma-ganzorzentropien}
For every $t \in \ut(\Sigma)$, $H_k^{\ell,\deg}(t) \leq H^{\deg}(t)$ and $H_k^{\deg,\ell}(t) \leq H_k^{\ell}(t)$ holds.
\end{lemma}
Moreover, one of the main results of \cite{Ganczorz20} states the following bounds:

\begin{theorem}[\mbox{\cite[Theorem~12]{Ganczorz20}}]\label{theorem-ganczorzentropybounds}
Let $t \in \ut(\Sigma)$, with $\sigma \leq |t|^{1-\alpha}$ for some $\alpha > 0$. Then $t$ can be represented within the following
bounds (in bits):
\begin{align*}
&H^{\deg}(t) + H_k^{\ell}(t) + \mathcal{O}\left(\frac{|t|k \log\sigma+|t|\log\log_{\sigma}|t|}{\log_{\sigma}|t|}\right),\\
&H_k^{\ell,\deg}(t) + H_k^{\ell}(t) + \mathcal{O}\left(\frac{|t|k \log\sigma+|t|\log\log_{\sigma}|t|}{\log_{\sigma}|t|}\right), \\
&H_k^{\deg,\ell}(t) + H^{\deg}(t) + \mathcal{O}\left(\frac{|t|k \log\sigma+|t|\log\log_{\sigma}|t|}{\log_{\sigma}|t|}\right) .
\end{align*}
\end{theorem}

\subsection{Label-shape entropy}
Another notion of empirical entropy for trees which incorporates both node labels and tree structure was recently introduced in \cite{HuckeLS19}: Let 
us start with a binary tree $t \in \bt(\Sigma)$. The $k^{th}$-order label-shape entropy $H_k(t)$ of $t$ (in \cite{HuckeLS19} it is simply called the $k^{th}$-order empirical entropy of $t$)
is defined as
\begin{align}\label{def-h_kbinary}
H_k(t) = \sum_{z \in \mathcal{L}_k}\sum_{\tilde{a} \in \Sigma \times \{0,2\}}m_{z,\tilde{a}}^t\log_2\left(\frac{m_z^t}{m_{z,\tilde{a}}^t}\right),
\end{align}
where $m_z^t$ and $m_{z,\tilde{a}}^t$ are from \eqref{m_z^t} and \eqref{m_z,tildea^t}, respectively.
Now let $t \in \ut(\Sigma)$ be an unranked tree and recall that $\ffcns(t) \in \bt(\Sigma)$. The $k^{th}$-order label-shape entropy $H_k(t)$ of $t$ is defined as
\begin{align}\label{def-h_kfcns}
H_k(t)=H_k(\ffcns(t)).
\end{align}
The following result is shown in \cite{HuckeLS19} using a grammar-based encoding of trees: 
\begin{theorem}\label{theorem-isitpaper}
Every tree $t \in \ut(\Sigma)$ can be represented within the following bound (in bits):
\begin{align*}
H_k(t) + \mathcal{O}\left(\frac{k|t|\log\sigma}{\log_{\sigma}|t|}\right)+\mathcal{O}\left(\frac{|t|\log\log_{\sigma}|t|}{\log_{\sigma}|t|}\right)+\sigma .
\end{align*}
\end{theorem}
Note that for binary trees, there are basically two possibilities how to compute the label-shape entropy $H_k(t)$: The first is to compute the label-shape entropy as defined in \eqref{def-h_kbinary}, the second is to consider the binary tree as an unranked tree and compute the label-shape entropy of its first-child next-sibling encoding as defined in \eqref{def-h_kfcns}. The following lemma, from \cite{Journalarxivversion} states that if we consider the first-child next-sibling encoding of the binary tree instead of the binary tree itself, the $k^{th}$-order label-shape entropy does not increase if we adapt the value of $k$ accordingly:

\begin{lemma}\label{Lemma-fcns}
Let $t \in \bt(\Sigma)$ denote a binary tree with first-child next-sibling encoding $\ffcns(t) \in \bt(\Sigma)$. Then $H_{2k}(\ffcns(t))\leq H_{k-1}(t)$ for $1 \leq k \leq n$.  
\end{lemma}
See \cite{Journalarxivversion} for a proof of Lemma~\ref{Lemma-fcns}. 
In contrast to Lemma~\ref{Lemma-fcns}, there are families of binary trees $t_n$ where $H_k(t_n) \in \Theta(n-k)$
and $H_k(\ffcns(t_n)) \in \Theta(\log(n-k))$ \cite{Journalarxivversion}.

\section{Comparison of the empirical entropy notions}

As we have seen in Theorems~\ref{theorem-ganczorzentropybounds} and~\ref{theorem-isitpaper}, entropy bounds for the number of bits needed to represent an unranked labeled tree $t$ are 
achievable by
\begin{itemize}
\item $H_k(t)$,
\item $H_k^{\ell}(t) + H_k^{\ell,\deg}(t)$,
\item $H^{\deg}(t)+H_k^{\deg,\ell}(t)$, and
\item $H^{\deg}(t)+H_k^{\ell}(t)$,
\end{itemize}
where in all cases we have to add a low-order term.
The term $H^{\deg}(t)+H_k^{\ell}(t)$ is lower-bounded by  $H_k^{\ell}(t) + H_k^{\ell,\deg}(t)$ and $H^{\deg}(t) + H_k^{\deg,\ell}(t)$ by Lemma~\ref{lemma-ganzorzentropien}. For the special case of unlabeled unranked trees, $H^{\deg}(t)$ (plus low-order terms) is an upper bound on the encoding length (see Theorem~\ref{theorem-degentropybound}). Thus, for the special case of unlabeled trees, we will also compare the entropy bounds to $H^{\deg}(t)$.

\subsection{Unlabeled binary trees}

In this subsection, we consider unlabeled binary trees, i.e., trees $t \in \mathcal{B}(\{a\})$ over the unary alphabet $\Sigma=\{a\}$. 
As $\Sigma = \{a\}$, the fixed dummy symbol used to pad $k$-histories and $k$-label-histories is $\Box = a$. We start with a simple lemma:

\begin{lemma}\label{lemma-degreebinary}
Let $t$ be an unlabeled binary tree with $n$ leaves ($|t|=2n-1$). Then $H^{\deg}(t) = H_k^{\ell,\deg}(t) = (2-o(1))n$.
\end{lemma}
\begin{proof}
Every binary tree of size $2n-1$ consists of $n$ nodes of degree $0$ and $n-1$ nodes of degree $2$ (independently of the shape of the binary tree). Thus, we obtain:
\begin{align*}
H^{\deg}(t) &= \sum_{i=0}^{|t|}\tildem_{i}^t\log_2\left(\frac{|t|}{\tildem_{i}^t}\right)=  n\log_2\left(\frac{2n-1}{n}\right) + (n-1)\log_2\left(\frac{2n-1}{n-1}\right) \\
&= (2n-1)g(n) \geq  2n(1-o(1)),
\end{align*}
where $g : [2,\infty) \to \mathbb{R}$ is the mapping defined by
$$
g(x) = \frac{x}{2x-1} \log_2\left(\frac{2x-1}{x}\right) +\frac{x-1}{2x-1}\log_2\left(\frac{2x-1}{x-1}\right) .
$$ 
It converges to $1$ from below for $x \to \infty$.
Moreover, as $t$ is unlabeled, every node has the same label and the same label-history. Thus, $H_k^{\ell, \deg}(t) = H^{\deg}(t)$. \qed
\end{proof}
On the other hand, for the label-shape entropy we have:

\begin{lemma}\label{lemma-vergleich3}
There exists a family of unlabeled binary trees $(t_n)_{n \geq 1}$ such that $|t_n|=2n-1$ and $H_k(t_n) \leq \log_2(en)$
for all $n \geq 1$ and $1 \leq k \leq n$.
\end{lemma}

\begin{proof}
We define $t_1 = a$ and $t_n = a(t_{n-1},a)$ for $n \geq 2$. 
Hence, $t_n$ is a left-degenerate binary tree with $n$ leaves such that every node is labeled with the symbol $a$.
Figure~\ref{fig-t6} shows $t_6$.
Fix an integer $k \geq 1$ and let $\Box = a$ be the fixed dummy 
symbol in $\Sigma$ used for padding histories shorter than $k$. 
For $n=1$, we have $H_k(t_1)=0$. Assume now that $n > 1$.
We start with computing the $k^{th}$-order label-shape entropy 
$H_k(t_n)$.  Only two $k$-histories appear in $t_n$:
\begin{itemize}
\item $z_0 = (a0)^k$: there are $n$ nodes with this history, namely one node with $\lambda$-value  
$(a,0)$ (the left most leaf) and $n-1$ nodes $v$ with $\lambda(v) = (a,2)$ (the $n-1$ internal nodes). 
\item $z_1 = (a0)^{k-1}a1$:  there are $n-1$ nodes with this history and all of them have the $\lambda$-value $(a,0)$.
\end{itemize}
Altogether, the $k^{th}$-order label-shape entropy of $t_n$ is
$$
H_k(t_n) = \log_2n + (n-1)\log_2\left(\frac{n}{n-1}\right) \leq \log_2n+\log_2e,
$$
where the last inequality follows from $(\frac{n}{n-1})^{n-1} \leq e$.
\qed
\end{proof}
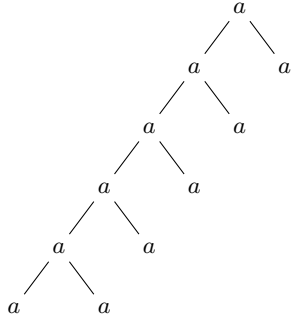
\begin{figure}[t]
\begin{center}
		\tikzset{level 1/.style={sibling distance=12mm}}
		\tikzset{level 2/.style={sibling distance=12mm}}
		\tikzset{level 3/.style={sibling distance=12mm}}  
		\tikzset{level 4/.style={sibling distance=12mm}}
		\begin{tikzpicture}[scale=1,auto,swap,level distance=8mm]
		\node (eps) {$a$} 
		child {node {$a$}
		  child {node{$a$} 
		  child {node {$a$}
		  child {node {$a$}
		 child {node {$a$}}
		child {node {$a$}}}
		child {node {$a$}}}
		child {node {$a$}}}
		child {node {$a$}}}
		child{node{$a$}};
		\end{tikzpicture}
\end{center}
	\caption{The binary tree $t_6$ from Lemma~\ref{lemma-vergleich3}.}	
	\label{fig-t6}
\end{figure}
Lemmas~\ref{lemma-degreebinary} and~\ref{lemma-vergleich3} already indicate that all entropies considered in this paper except for the label-shape entropy
are not interesting for unlabeled binary trees. For every unlabeled binary tree $t$ with $n$ leaves (and $2n-1$ nodes) we have:
\begin{itemize}
\item $H_k^{\ell}(t) = H_k^{\deg,\ell}= 0$, as every node  of $t$ has the same label.
\item $H_k^{\ell}(t)+H_k^{\ell,\deg}(t)=H^{\deg}(t)+H_k^{\deg,\ell}(t)=H_k^{\ell}(t) + H^{\deg}(t)=H^{\deg}(t)$ and these values are lower bounded
by $2n(1-o(1))$ (Lemma~\ref{lemma-degreebinary}).
\end{itemize}
The only notion of empirical tree entropy that is able to capture regularities in unlabeled binary trees (and that attains different values for different binary trees of the same size) is the label-shape entropy \eqref{def-h_kbinary} from \cite{HuckeLS19}.

\subsection{Labeled binary trees}

Next, we consider binary trees $t \in \bt(\Sigma)$, where $\Sigma$ is arbitrary. By Lemmas~\ref{lemma-degreebinary} and~\ref{lemma-vergleich3}, we already know that there are 
families $(t_n)_{n \geq 1}$ of binary trees, for which $t_n$ has $n$ leaves and $H_k(t_n)$ is exponentially smaller than $H^{\deg}(t_n)+H_k^{\deg,\ell}(t_n)$ and $H_k^{\ell}(t_n)+H_k^{\ell,\deg}(t_n)$ (and thus, $H_k^{\ell}(t_n)+H^{\deg}(t)$). 
As in the special case of unlabeled binary trees, we find that $H^{\deg}(t) = 2n(1-o(1))$ for every binary tree $t$ of size $2n-1$ (the node labels do not influence
$H^{\deg}(t)$), which implies $H^{\deg}(t) + H_k^{\deg,\ell}(t) \geq 2n(1-o(1))$.
The following lemma shows that $H_k(t)$ is always smaller than $H_k^{\ell}(t) + H_k^{\ell,\deg}(t)$ and $H^{\deg}(t) + H_k^{\deg,\ell}(t)$ (and  hence also 
$H_k^{\ell}(t)+H^{\deg}(t)$) for $t \in \bt(\Sigma)$:

\begin{lemma}
Let $t \in \bt(\Sigma)$ be a binary tree. Then 
\begin{enumerate}[(i)]
\item $H_k(t) \leq H_k^{\ell}(t) + H_k^{\ell,\deg}(t)$ and
\item $H_k(t) \leq H^{\deg}(t) + H_k^{\deg,\ell}(t)$.
\end{enumerate}
\end{lemma}

\begin{proof}
We start with proving statement (i): We have
\begin{align*}
& H_k(t) = \sum_{z \in \mathcal{L}_k}\sum_{a \in \Sigma}\sum_{i \in \{0,2\}}m_{z,(a,i)}^t\log_2\left(\frac{m_z^t}{m_{z,(a,i)}^t}\right)\\
&= \sum_{z \in \mathcal{L}_k}\sum_{a \in \Sigma}\sum_{i \in \{0,2\}}m_{z,(a,i)}^t\left(\log_2\left(\frac{m_z^t}{m_{z,(a,0)}^t+m_{z,(a,2)}^t}\right) + \log_2\left(\frac{m_{z,(a,0)}^t+m_{z,(a,2)}^t}{m_{z,(a,i)}^t}\right)\right)\\
& =  \sum_{z \in \mathcal{L}_k}\sum_{a \in \Sigma}\left(m_{z,(a,0)}^t+m_{z,(a,2)}^t\right)
\log_2\left(\frac{m_z^t}{m_{z,(a,0)}^t+m_{z,(a,2)}^t}\right)  \\
&+
\sum_{z \in \mathcal{L}_k}\sum_{a \in \Sigma}\sum_{i \in \{0,2\}}m_{z,(a,i)}^t  \log_2\left(\frac{m_{z,(a,0)}^t+m_{z,(a,2)}^t}{m_{z,(a,i)}^t}\right)\\
&\leq  \sum_{z \in \Sigma^k}\sum_{a \in \Sigma}\tildem_{z,a}^t
 \log_2\left(\frac{\tildem_z^t}{\tildem_{z,a}^t}\right)
+
\sum_{z \in \Sigma^k}\sum_{a \in \Sigma}\sum_{i \in \{0,2\}}\tildem_{z,i,a}^t  \log_2\left(\frac{\tildem_{z,a}^t}{\tildem_{z,i,a}^t}\right)\\
&= H_k^{\ell}(t) + H_k^{\ell, \deg}(t),
\end{align*}
where the inequality  in the second last line follows from the log-sum inequality (Lemma~\ref{logsum})
and the last equality follows from the fact that in a binary tree, every node is either of degree $0$ or $2$. Statement (ii) can be shown in a similar way:
\begin{align*}
&H_k(t) = \sum_{z \in \mathcal{L}_k}\sum_{a \in \Sigma}\sum_{i \in \{0,2\}}m_{z,(a,i)}^t\log_2\left(\frac{m_z^t}{m_{z,(a,i)}^t}\right)\\
&=\sum_{z \in \mathcal{L}_k}\sum_{a \in \Sigma}\sum_{i \in \{0,2\}}m_{z,(a,i)}^t\left(\log_2\left(\frac{m_z^t}{\sum_{a \in \Sigma}m_{z,(a,i)}^t}\right) + \log_2\left(\frac{\sum_{a \in \Sigma}m_{z,(a,i)}^t}{m_{z,(a,i)}^t}\right)\right)\\
& =\sum_{z \in \mathcal{L}_k}\sum_{i \in \{0,2\}}\left(\sum_{a \in \Sigma}m_{z,(a,i)}^t\right)
\log_2\left(\frac{m_z^t}{\sum_{a \in \Sigma}m_{z,(a,i)}^t}\right)\\
&+
\sum_{z \in \mathcal{L}_k}\sum_{a \in \Sigma}\sum_{i \in \{0,2\}}m_{z,(a,i)}^t  \log_2\left(\frac{\sum_{a \in \Sigma}m_{z,(a,i)}^t}{m_{z,(a,i)}^t}\right)\\
&\leq \sum_{i \in \{0,2\}}n_i^t
 \log_2\left(\frac{|t|}{n_i^t}\right)
+
\sum_{z \in \Sigma^k}\sum_{a \in \Sigma}\sum_{i \in \{0,2\}}\tildem_{z,i,a}^t  \log_2\left(\frac{\tildem_{z,i}^t}{\tildem_{z,i,a}^t}\right)\\
&= H^{\deg}(t) + H_k^{\deg,\ell}(t),
\end{align*}
where the inequality  follows again from the log-sum inequality.
\qed\end{proof}

\subsection{Unlabeled unranked trees}

In this subsection, we consider unranked trees $t \in \ut(\Sigma)$ over the unary alphabet $\Sigma=\{a\}$. As $\Sigma = \{a\}$, the fixed dummy symbol used to pad $k$-histories and $k$-label-histories is $\Box = a$. Moreover, note that in order to compute $H_k(t)$ for an unranked tree $t \in \ut(\Sigma)$, we have to consider $\ffcns(t)$. Note that $\ffcns(t)$ is then an unlabeled binary tree: we must take $\Box = a$ by our conventions for the dummy symbol; hence the fresh $\Box$-labeled leaves in $\ffcns(t)$ are labeled
with $a$, too. 

As in the case of unlabeled binary trees, we observe that some entropy measures, in particular those that involve labels, only attain trivial values for unranked unlabeled trees.
More precisely, for every tree
$t \in \ut(\{a\})$ we have 
\begin{itemize}
\item  $H_k^{\ell}(t)=H_k^{\deg,\ell}(t)=0$, as every node has the same label $a$, and
\item $H^{\deg}(t) = H_k^{\ell,\deg}(t)$, as every node has the same $k$-label-history and the same label.
\item  We get $H_k^{\ell}(t) + H_k^{\ell,\deg}(t) = H^{\deg}(t) + H_k^{\deg,\ell}(t) = H^{\deg}(t)+ H_k^{\ell}(t) = H^{\deg}(t)$.
\end{itemize}
By this observation, we only compare $H_k(t)$ with $H^{\deg}(t)$ for $t \in \ut(\{a\})$ in this subsection. By Lemmas~\ref{lemma-degreebinary} and~\ref{lemma-vergleich3}, 
there exists a family of unlabeled trees $(t_n)_{n \geq 1}$ such that $|t_n|=\Theta(n)$ and for which $H_k(t_n)$ is exponentially smaller than $H^{\deg}(t_n)$.
For general unranked unlabeled trees, we find the following:

\begin{theorem}\label{theo-degreeentropy}
For every unlabeled unranked tree $t$ with $|t |\geq 2$ and integer $k \geq 1$, we have
$H_k(t) \leq 2H^{\deg}(t) + 2\log_2(|t|)+4$.
\end{theorem}
\begin{proof}
We start the proof with some simple counting facts for fcns-encodings.
Consider an unranked tree $t \in \ut(\{a\})$ with $|t |\geq 2$.
We claim that
\begin{enumerate}[(i)]
\item the number of inner nodes of $\ffcns(t)$ which are left children equals the number of nodes of $t$ of degree at least $1$, and
\item the number of leaves of $\ffcns(t)$ which are left children equals the number of nodes of $t$ which are leaves.
\end{enumerate} 
To show this, one should think of $\ffcns(t)$ as the tree obtained by taking all nodes of $t$ (and adding some fresh nodes as leaves). For a node $v \in V(t)$
its left (right) child in $\ffcns(t)$ is the first child (right sibling) of $v$ in $t$ if it exists. If it does not exist, we take 
a fresh leaf as the left (right) child of $v$ in $\ffcns(t)$. Then, the inner nodes of $\ffcns(t)$ are exactly the nodes of $t$.
The inner nodes of $\ffcns(t)$ are moreover in bijective correspondence with the nodes of $\ffcns(t)$ that are left children; the 
corresponding bijection is of course the function $\parent(\cdot)$ that maps a left child to its parent node.
Hence, $\parent(\cdot)$ can be viewed as a bijection from the left children in $\ffcns(t)$ to the nodes of $t$.
Consider a left child $v$ in $\ffcns(t)$ and let $v' = \parent(v)$ be the corresponding node in $t$. 
If $v$ is an inner node of $\ffcns(t)$ then $v'$ has a first child in $t$, i.e.,
its degree is at least one. On the other hand, if $v$ is a leaf of $\ffcns(t)$ then $v'$ has no first child in $t$, i.e.,
its degree is zero. This yields the above statements (i) and (ii).

Let us now fix $k \geq 1$ and let  
\begin{eqnarray*}
\mathcal{L}_k^0 &=& \{a i_1\cdots a i_{k-1} a 0 \mid i_1, \ldots, i_{k-1} \in \{0,1\} \} \subseteq  \mathcal{L}_k \text{ and} \\
\mathcal{L}_k^{1} &=& \{a i_1\cdots a i_{k-1} a 1 \mid i_1, \ldots, i_{k-1} \in \{0,1\} \} \subseteq  \mathcal{L}_k .  
\end{eqnarray*}
Let $n_{\geq 1}^t$ denote the number of nodes of $t$ of degree at least $1$ and for $z \in \mathcal{L}_k$ and $i \in \{0,2\}$ 
let $m_{z,i}^{\ffcns(t)}$ denote the number of nodes in $\ffcns(t)$ having $k$-history $z$
and degree $i$. From (i) and (ii) we get
\begin{align}\label{eq-number of nodes}
n_0^t = \sum_{z \in \mathcal{L}_k^0}m_{z,0}^{\ffcns(t)} \quad\text{ and }\quad n_{\geq 1}^t+1 = \sum_{z \in \mathcal{L}_k^0}m_{z,2}^{\ffcns(t)}.
\end{align}
The $+1$ in the second identity comes from the fact that on the right-hand side we also count the root node (which is not a left child of $\ffcns(t)$).
Thus, we have
\begin{eqnarray}
H^{\deg}(t) &=& \sum_{i=0}^{|t|} n_i^t \log_2 \left(\frac{|t|}{n_i^t}\right) \nonumber \\
&\geq &n_{\geq 1}^t\log_2\left(\frac{|t|}{n_{\geq 1}^t}\right) + n_{0}^t\log_2\left(\frac{|t|}{n_{0}^t}\right)  \nonumber  \\
&\geq &(n_{\geq 1}^t+1)\log_2\left(\frac{|t|+1}{n_{\geq 1}^t+1}\right) -\log_2 |t|  + \label{term1} \\
&& n_{0}^t\log_2\left(\frac{|t|+1}{n_{0}^t}\right) -\frac{n_0^t}{\ln(2)|t|},  \label{term2}
\end{eqnarray}
where for the last inequality, we used $y/x \geq (y+1)/(x+1)$ if $y \geq x$ to get the term \eqref{term1} and 
$$\log_2(|t|+1) - \log_2 |t| \leq 
\frac{|t|+1-|t|}{\ln (2) |t|}  = \frac{1}{\ln(2) |t|}
$$
to get the term \eqref{term2}. The inequality in the last line follows from the mean 
value theorem. Hence, by the above equations \eqref{eq-number of nodes} and the fact that $|t|+1$ equals the number of nodes $v$ of $\ffcns(t)$ with $k$-history $h_k(v) \in \mathcal{L}_k^0$, we get
\begin{eqnarray*}
H^{\deg}(t)
&\geq & \left(\sum_{z \in \mathcal{L}_k^0}m_{z,2} ^{\ffcns(t)}\right)\log_2\left(\frac{\sum_{z \in \mathcal{L}_k^0}m_z^{\ffcns(t)}}{\sum_{z \in \mathcal{L}_k^0}m_{z,2} ^{\ffcns(t)}}\right)+ \\
& & \left(\sum_{z \in \mathcal{L}_k^0}m_{z,0}^{\ffcns(t)}\right)\log_2\left(\frac{\sum_{z \in \mathcal{L}_k^0}m_z^{\ffcns(t)}}{\sum_{z \in \mathcal{L}_k^0}m_{z,0}^{\ffcns(t)}}\right)
 -\log_2|t|-2\\
& \geq & \sum_{z \in \mathcal{L}_k^0} \sum_{i\in\{0,2\}} m_{z,i}^{\ffcns(t)} \log_2\left(\frac{m_z^{\ffcns(t)}}{m_{z,i}^{\ffcns(t)}}\right)
 -\log_2|t|-2,
\end{eqnarray*}
where the last inequality follows from the log-sum inequality (Lemma~\ref{logsum}).

In the next part of the proof, we establish a similar estimate by considering nodes of $\ffcns(t)$ with $h_k(v) \in \mathcal{L}_k^{1}$. These nodes
are exactly the right children in $\ffcns(t)$ and there are $|t|$ many such nodes. The parent-mapping yields a bijection from the right children in $\ffcns(t)$ to the nodes of $t$. Consider a node $v$ in $\ffcns(t)$ and assume that $v$ is the right child of $v' = \parent(v)$.
If $v$ is a leaf of $\ffcns(t)$ then $v'$ does not have a right sibling in $t$ and
if $v$ is an inner node of $\ffcns(t)$ then $v'$ has a right sibling in $t$.
Hence, the number of leaves $v$ of $\ffcns(t)$ with $h_k(v) \in \mathcal{L}_k^1$
is equal to the number of nodes in $t$ that do not have a right sibling. 
There are exactly $n_{\geq 1}^t+1$ such nodes (there are $n_{\geq 1}^t$ nodes that are the right-most
child of their parent node; in addition the root has no right sibling too).
Hence, we get:
\begin{itemize}
\item[(iii)] The number of leaves $v$ of $\ffcns(t)$ with $h_k(v) \in \mathcal{L}_k^1$ equals one plus the number of nodes of $t$ of degree at least $1$:
\begin{align*}
\sum_{z \in \mathcal{L}_k^1}m_{z,0}^{\ffcns(t)} = n_{\geq 1}^t+1.
\end{align*}
\item[(iv)] For the number of leaves of $|t|$, we thus obtain:
\begin{align*}
n_0^t &= |t|-\sum_{z \in \mathcal{L}_k^1}m_{z,0}^{\ffcns(t)} +1= \sum_{z \in \mathcal{L}_k^1}m_{z}^{\ffcns(t)}-\sum_{z \in \mathcal{L}_k^1}m_{z,0}^{\ffcns(t)}+1\\
&=\sum_{z \in \mathcal{L}_k^1}m_{z,2} ^{\ffcns(t)}+1.
\end{align*}
\end{itemize}
We thus find
\begin{eqnarray*}
H^{\deg}(t) &=& \sum_{i=0}^n n_i^t \log_2\left(\frac{|t|}{n_i^t}\right)\\
&\geq & n_{\geq 1}^t\log_2\left(\frac{|t|}{n_{\geq 1}^t}\right)+n_{0}^t\log_2\left(\frac{|t|}{n_{0}^t}\right)\\
&\geq & (n_{\geq 1}^t+1)\log_2\left(\frac{|t|}{n_{\geq 1}^t+1}\right) + \\
& & (n_0^t-1)\log_2\left(\frac{|t|}{n_0^t-1}\right)-\log_2(|t|)-2,
\end{eqnarray*}
where the last estimate follows from the fact that the mapping $x \mapsto H(x) - H(x+1)$ 
($x \in [0,|t|-1]$) 
with 
$$
H(x) = x\log_2\left(\frac{|t|}{x}\right)+(|t|-x)\log_2\left(\frac{|t|}{|t|-x}\right)
$$
the binary entropy function is minimal for $x=0$
and $H(0) - H(1) = - \log_2(|t|) - (|t|-1) \log_2\left(\frac{|t|}{|t|-1}\right) \geq - \log_2(|t|) - \log_2(e)$.
By the above equations in (iii) and (iv), we thus get
\begin{eqnarray*}
H^{\deg}(t)&\geq& \left( \sum_{z \in \mathcal{L}_k^1}m_{z,0}^{\ffcns(t)}\right) \log_2\left(\frac{\sum_{z \in \mathcal{L}_k^1}m_{z}^{\ffcns(t)} }{\sum_{z \in \mathcal{L}_k^1}m_{z,0}^{\ffcns(t)}}\right) \\
&+&\left(\sum_{z \in \mathcal{L}_k^1}m_{z,2} ^{\ffcns(t)}\right) \log_2\left(\frac{\sum_{z \in \mathcal{L}_k^1}m_{z}^{\ffcns(t)} }{\sum_{z \in \mathcal{L}_k^1}m_{z,2} ^{\ffcns(t)}}\right)-\log_2(|t|)-2\\
&\geq & \sum_{z \in \mathcal{L}_k^1} \sum_{i\in\{0,2\}} m_{z,i} ^{\ffcns(t)} \log_2\left(\frac{m_{z}^{\ffcns(t)} }{m_{z,i} ^{\ffcns(t)}}\right)-\log_2(|t|)-2, \\
\end{eqnarray*}
where the last inequality follows from the log-sum inequality.
Altogether, since $\mathcal{L}_k$ is the disjoint union of $\mathcal{L}^0_k$ and $ \mathcal{L}^1_k$, we obtain:
\begin{equation*}
H_k(t) = \sum_{z \in \mathcal{L}_k}    \sum_{i\in \{0,2\}} m_{z,i}^{\ffcns(t)} \log_2\left(\frac{m_{z}^{\ffcns(t)} }{m_{z,i} ^{\ffcns(t)}}\right) 
\leq  2H^{\deg}(t) + 2\log_2(|t|)+4.
\end{equation*}
This proves the theorem.
\qed\end{proof}
Moreover, as  $H^{\deg}(t) = H_k^{\ell}(t) +H_k^{\deg,\ell}(t)=H^{\deg}(t)+H_k^{\ell,\deg}(t)$ for every tree $t \in \ut(\{a\})$ and $k \geq 0$, we obtain the following corollary from Theorem~\ref{theo-degreeentropy}:

\begin{corollary}
For every unlabeled unranked tree $t \in \ut(\{a\})$ with $|t|\geq 2$ and integer $k\geq 1$, we have 
\begin{itemize}
\item $H_k(t)\leq 2(H^{\deg}(t)+H_k^{\deg,\ell}(t))+2\log_2(|t|)+4$, and
\item $H_k(t)\leq 2(H_k^{\ell,\deg}(t) + H_k^{\ell}(t))+2\log_2(|t|)+4$.
\end{itemize}
\end{corollary}
It remains to remark that if we consider unranked trees over an alphabet $\Sigma$ of size $\sigma>1$, there are examples of families of trees, for which the degree entropy is asymptotically exponentially smaller than the $k^{th}$-order label-shape tree entropy. This is not very surprising as the label-shape entropy incorporates the node labels, while the degree entropy does not.

\subsection{Labeled unranked trees}

In this subsection, we consider general unranked labeled trees $t \in \ut(\Sigma)$ over alphabets $\Sigma$ of arbitrary size. The entropies to be compared in this general case are $H_k(t)$, $H^{\deg}(t)+H_k^{\deg,\ell}(t)$,  $H_k^{\ell,\deg}(t) + H_k^{\ell}(t)$ and $H^{\deg}(t) + H_k^{\ell}(t)$. 
Somewhat surprisingly it turns out that $H_k^{\ell}(t) + H_k^{\ell,\deg}(t)$ is always upper-bounded by $H^{\deg}(t)+H_k^{\deg,\ell}(t)$:

\begin{theorem}\label{thm-factor-2}
Let $t \in \ut(\Sigma)$. Then $H_k^{\ell}(t) + H_k^{\ell,\deg}(t) \leq H^{\deg}(t)+H_k^{\deg,\ell}(t)$.
\end{theorem}
\begin{proof}
We have
\begin{align*}
H_k^{\ell}(t)+H_k^{\ell, \deg}(t) &= \sum_{z \in \Sigma^k}\sum_{a \in \Sigma}n_{z,a}^t\log_2\left(\frac{n_z^t}{n_{z,a}^t}\right) + \sum_{z \in \Sigma^k}\sum_{a \in \Sigma}\sum_{i=0}^{|t|} n_{z,i,a}^t\log_2\left(\frac{n_{z,a}^t}{n_{z,i,a}^t}\right)\\
&= \sum_{z \in \Sigma^k}\sum_{a \in \Sigma}\sum_{i=0}^{|t|} n_{z,i,a}^t\log_2\left(\frac{n_z^t}{n_{z,a}^t}\right) + \sum_{z \in \Sigma^k}\sum_{a \in \Sigma}\sum_{i=0}^{|t|} n_{z,i,a}^t\log_2\left(\frac{n_{z,a}^t}{n_{z,i,a}^t}\right)\\
&= \sum_{z \in \Sigma^k}\sum_{a \in \Sigma}\sum_{i=0}^{|t|}n_{z,i,a}^t\log_2\left(\frac{n_z^t}{n_{z,i,a}^t}\right)\\
&= \sum_{z \in \Sigma^k}\sum_{a \in \Sigma}\sum_{i=0}^{|t|} n_{z,i,a}^t\log_2\left(\frac{n_z^t}{n_{z,i}^t}\right) + \sum_{z \in \Sigma^k}\sum_{a \in \Sigma}\sum_{i=0}^{|t|} n_{z,i,a}^t\log_2\left(\frac{n_{z,i}^t}{n_{z,i,a}^t}\right)\\
&=\sum_{z \in \Sigma^k}\sum_{i=0}^{|t|} n_{z,i}^t\log_2\left(\frac{n_z^t}{n_{z,i}^t}\right) + \sum_{z \in \Sigma^k}\sum_{a \in \Sigma}\sum_{i=0}^{|t|} n_{z,i,a}^t\log_2\left(\frac{n_{z,i}^t}{n_{z,i,a}^t}\right)\\
&\leq H^{\deg}(t) + H_k^{\deg, \ell}(t),
\end{align*}
where the inequality in the last line follows from the log-sum inequality (Lemma~\ref{logsum}).
This proves the theorem.
\end{proof}
As a corollary of Theorem~\ref{thm-factor-2} it turns out that 
$H^{\deg}(t)+H_k^{\deg,\ell}(t)$ and $H_k^{\ell}(t) + H^{\deg}(t)$ are equivalent up to a constant factor.

\begin{corollary} 
Let $t \in \ut(\Sigma)$. Then 
$$H^{\deg}(t)+H_k^{\deg,\ell}(t) \leq  H^{\deg}(t) + H_k^{\ell}(t) \leq 2H^{\deg}(t)+H_k^{\deg,\ell}(t).$$
\end{corollary}

\begin{proof}
The first inequality follows from Lemma~\ref{lemma-ganzorzentropien}.
By Theorem \ref{thm-factor-2}, we have $H_k^{\ell}(t) \leq H^{\deg}(t)+H_k^{\deg,\ell}(t)$ from which the statement follows.
\qed\end{proof}
In the rest of the section we present three examples showing that in all cases that are not covered
by Theorem~\ref{thm-factor-2} we can achieve a non-constant (in most cases even exponential) separation between the corresponding entropies.

\begin{lemma}\label{lemma-vergleich1}
There exists a family of unranked trees $(t_n)_{n \geq 1}$ such that for all $n \geq 1$
and $1 \leq k \leq 2n$:
\begin{enumerate}[(i)]
\item $|t_n|=2n+1$,
\item $H_k(t_n)\leq \log_2(e)+ \log_2\left(n-\left\lfloor \frac{k-1}{2} \right\rfloor\right)+2$,
\item  $H_k^{\deg,\ell}(t_n) = 2n$ and hence $H^{\deg}(t_n)+H_k^{\deg,\ell}(t_n) \geq 2n$, and 
\item $H_k^{\ell}(t_n) \geq 2n$ and hence $H_k^{\ell}(t_n) + H_k^{\ell,\deg}(t_n) \geq 2n$.
\end{enumerate}
\end{lemma}

\begin{proof}
Define the unranked tree $t_n$ as $t_n = a((bc)^n)$, that is, $t_n$ is a tree consisting of a root node of degree $2n$ labeled with $a$ and $2n$ leaves, of which $n$ many leaves are labeled $b$ and $n$ many leaves are labeled $c$. The tree $t_3$ is depicted in Figure~\ref{fig-t5} on the left.
First, we compute the degree-label entropy of $t_n$: Let $\Box =a$ denote the fixed dummy symbol used to pad histories shorter than $k$. For every node $v$ of $t_n$, we have $\lh_k(v)=a^k$, which yields
\begin{equation*}
H_k^{\deg,\ell}(t_n) = \sum_{z \in \Sigma^k}\sum_{i=0}^{|t|} \sum_{a \in \Sigma}\tildem_{z,i,a}^t\log_2\left(\frac{\tildem_{z,i}^t}{\tildem_{z,a,i}^t}\right) =\log_21 + n\log_22+ n\log_22 = 2n.
\end{equation*}
This shows statement (iii). Moreover, (iv) follows from $H_k^{\deg,\ell}(t_n) = 2n$ and Lemma~\ref{lemma-ganzorzentropien}.

It remains to compute the $k^{th}$-order label-shape entropy of $t_n$: 
For this, we have to consider the first-child next-sibling encoding of $t_n$. Let $\Box=a$ denote the dummy symbol labeling the leaves of $\ffcns(t_n)$
as well as the dummy symbol used to pad histories shorter than $k$. 
The tree $\ffcns(t_3)$ is depicted in Figure~\ref{fig-t5} on the right.
Intuitively, most $k$-histories determine the $\lambda$-value of the corresponding node, which leads to a low
$k^{th}$-order label-shape entropy. Formally, consider the $k$-history $(b1c1)^{k/2}$, if $k$ is even, respectively, $(c1b1)^{(k-1)/2}c1$, if $k$ is odd. There are $n-\lfloor (k-1)/2 \rfloor$ many nodes of this history in $\ffcns(t_n)$, and $n-\lfloor (k-1)/2 \rfloor-1$ many of them are inner nodes labeled $b$ while one of them is a leaf labeled $\Box=a$. Furthermore, consider the $k$-history $(a0)^k$. There are two nodes of this $k$-history, one of them labeled $a$ (the root node) and one of them labeled $b$ (the left child of the root node). For all other $k$-histories $z$ occurring in $\ffcns(t_n)$, we find that all nodes with $k$-history $z$ have the same  $\lambda$-value. Thus, we have
\begin{eqnarray*}
H_k(t_n) &=& \left(n-\left\lfloor \frac{k-1}{2} \right\rfloor-1\right)\log_2\left(\frac{n-\lfloor \frac{k-1}{2} \rfloor}{n-\lfloor \frac{k-1}{2} \rfloor-1}\right)+\log_2\left(n-\left\lfloor \frac{k-1}{2} \right\rfloor\right)+2\\
&\leq & \log_2(e)+ \log_2\left(n-\left\lfloor \frac{k-1}{2} \right\rfloor\right)+2.
\end{eqnarray*}
This shows statement (ii).
\qed\end{proof}
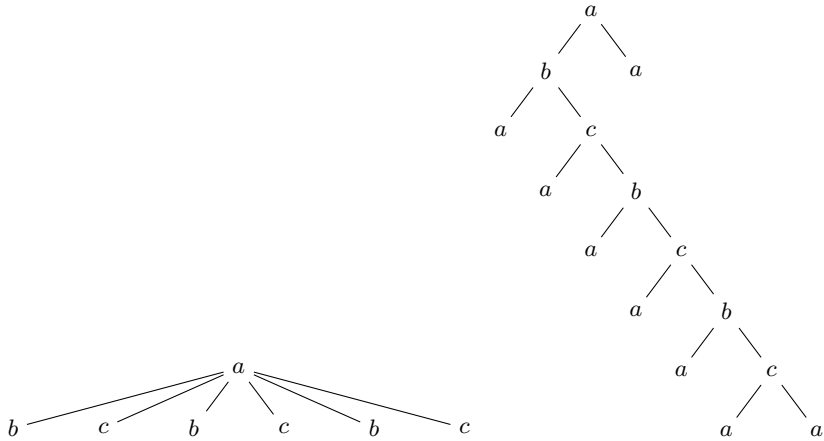
\begin{figure}[t]
\begin{center}
		\tikzset{level 1/.style={sibling distance=12mm}}
		\tikzset{level 2/.style={sibling distance=12mm}}
		\tikzset{level 3/.style={sibling distance=12mm}}  
		\tikzset{level 4/.style={sibling distance=12mm}}
		\begin{tikzpicture}[scale=1,auto,swap,level distance=8mm]
		\node (eps) {$a$} 
		child {node {$b$}}
		  child {node{$c$} }
		child {node {$b$}}
		child {node {$c$}}
		child {node {$b$}}
		child{node{$c$}};
		\end{tikzpicture}
		\begin{tikzpicture}[scale=1,auto,swap,level distance=8mm]
		\node (eps) {$a$} 
		child {node {$b$}
		  child {node{$a$} }
		  child {node{$c$}
		  child {node{$a$} }
		  child {node {$b$}
		  child {node{$a$} }
		  child {node{$c$}
		  child {node{$a$} }
		  child {node {$b$}
		  child {node{$a$} }
		  child {node{$c$}
		  child {node{$a$} }
		  child{node{$a$}}}}}}}}
		child{node{$a$}};
		\end{tikzpicture}
	\end{center}
	\caption{The binary tree $t_3$ from Lemma~\ref{lemma-vergleich1} (left) and its first-child next-sibling encoding $\ffcns(t_3)$ (right).}	
	\label{fig-t5}
\end{figure}
Lemma~\ref{lemma-vergleich1} shows that there are not only families of binary trees, but also families of unranked (non-binary) trees $(t_n)_{n \geq 1}$ (for which we have to compute $H_k(t_n)$ via the $\ffcns$-endcoding) such that $|t_n|=\Theta(n)$ and $H_k(t_n)$ is exponentially smaller than $H^{\deg}(t_n)+H_k^{\deg,\ell}(t_n)$ and $H_k^{\ell,\deg}(t_n) + H_k^{\ell}(t_n)$. The next lemma shows that there are also families of trees $(t_n)_{n \geq 1}$ such that $H_k^{\ell,\deg}(t_n) +H_k^{\ell}(t_n)$ is (even more than) exponentially smaller than $H^{\deg}(t_n) + H_k^{\deg,\ell}(t_n)$ (and thus, than $H^{\deg}(t_n)+ H_k^{\ell}(t_n)$) and $H_k(t_n)$:

\begin{lemma}\label{lemma-vergleich2}
There exists a family of unranked trees $(t_n)_{n \geq 1}$ such that for all $n \geq 1$ and $1 \leq k \leq n$:
\begin{enumerate}[(i)]
\item $|t_n|=3n+3$,
\item  $H_k(t_n) \geq 2(n-k+1)$,
\item $H^{\deg}(t_n) + H_k^{\deg,\ell}(t_n)\geq 2 n$ and
\item $H_k^{\ell}(t_n)+H_k^{\ell,\deg}(t_n)=3\log_2(3)$.
\end{enumerate}
\end{lemma}

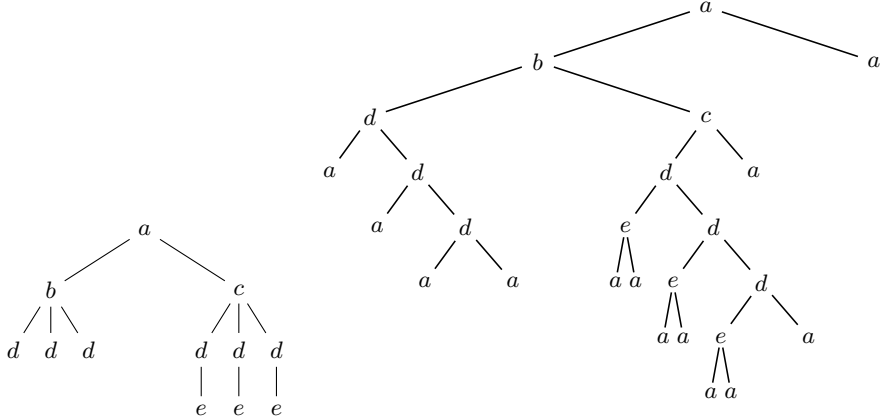
\begin{figure}[t]
\begin{center}
		\tikzset{level 1/.style={sibling distance=25mm}}
		\tikzset{level 2/.style={sibling distance=5mm}}
		\tikzset{level 3/.style={sibling distance=5mm}}  
		\tikzset{level 4/.style={sibling distance=12mm}}
		\begin{tikzpicture}[scale=1,auto,swap,level distance=8mm]
		\node (eps) {$a$} 
		child {node {$b$}
		child {node {$d$}}
		child {node {$d$}}
		child {node {$d$}}}
		  child {node{$c$} 
		  child {node {$d$}
		  child {node {$e$}}}
		  child {node {$d$}
		  child {node {$e$}}}
		  child {node {$d$}
		  child {node {$e$}}}};
		\end{tikzpicture}
		\begin{tikzpicture}[-,>=stealth',shorten <=-5.5pt, shorten >=-5.5pt, auto, node distance=2em and 0.5em, semithick]
  \tikzstyle{every state}=[]
\node[state, draw=none](a)[]{$a$};
\node[state, draw=none](n1)[below=-0.1cm of a]{};
\node[state, draw=none](b)[left=1.4cm of n1]{$b$};
\node[state, draw=none](a1)[right=1.4cm of n1]{$a$};
\node[state, draw=none](n2)[below=-0.1cm of b]{};
\node[state, draw=none](bd1)[left=1.4cm of n2]{$d$};
\node[state, draw=none](c)[right=1.4cm of n2]{$c$};
\node[state, draw=none](n3)[below=-0.1cm of bd1]{};
\node[state, draw=none](a2)[left=-0.3cm of n3]{$a$};
\node[state, draw=none](bd2)[right=-0.2cm of n3]{$d$};
\node[state, draw=none](n4)[below=-0.1cm of bd2]{};
\node[state, draw=none](a3)[left=-0.3cm of n4]{$a$};
\node[state, draw=none](bd3)[right=-0.2cm of n4]{$d$};
\node[state, draw=none](n4)[below=-0.1cm of bd3]{};
\node[state, draw=none](a4)[left=-0.3cm of n4]{$a$};
\node[state, draw=none](bd4)[right=-0.2cm of n4]{$a$};

\node[state, draw=none](n7)[below=-0.1cm of c]{};
\node[state, draw=none](cd1)[left=-0.3cm of n7]{$d$};
\node[state, draw=none](a8)[right=-0.2cm of n7]{$a$};

\node[state, draw=none](n8)[below=-0.1cm of cd1]{};
\node[state, draw=none](a9)[left=-0.3cm of n8]{$e$};
\node[state, draw=none](m7)[below=-0.1cm of a9]{};
\node[state, draw=none](a36)[left=-0.7cm of m7]{$a$};
\node[state, draw=none](a37)[right=-0.7cm of m7]{$a$};

\node[state, draw=none](cd2)[right=-0.2cm of n8]{$d$};
\node[state, draw=none](n9)[below=-0.1cm of cd2]{};
\node[state, draw=none](a10)[left=-0.3cm of n9]{$e$};
\node[state, draw=none](m8)[below=-0.1cm of a10]{};
\node[state, draw=none](a46)[left=-0.7cm of m8]{$a$};
\node[state, draw=none](a47)[right=-0.7cm of m8]{$a$};

\node[state, draw=none](cd3)[right=-0.2cm of n9]{$d$};
\node[state, draw=none](n10)[below=-0.1cm of cd3]{};
\node[state, draw=none](a11)[left=-0.3cm of n10]{$e$};
\node[state, draw=none](m9)[below=-0.1cm of a11]{};
\node[state, draw=none](a56)[left=-0.7cm of m9]{$a$};
\node[state, draw=none](a57)[right=-0.7cm of m9]{$a$};

\node[state, draw=none](cd4)[right=-0.2cm of n10]{$a$};

\path (a) edge node{} (b);
\path (a) edge node{} (a1);

\path (b) edge node{} (bd1);
\path (b) edge node{} (c);

\path (bd1) edge node{} (bd2);
\path (bd1) edge node{} (a2);

\path (bd2) edge node{} (bd3);
\path (bd2) edge node{} (a3);

\path (bd3) edge node{} (bd4);
\path (bd3) edge node{} (a4);

\path (c) edge node{} (cd1);
\path (c) edge node{} (a8);

\path (cd1) edge node{} (cd2);
\path (cd1) edge node{} (a9);

\path (cd2) edge node{} (cd3);
\path (cd2) edge node{} (a10);

\path (cd3) edge node{} (cd4);
\path (cd3) edge node{} (a11);

\path (a9) edge node{} (a36);
\path (a9) edge node{} (a37);

\path (a10) edge node{} (a46);
\path (a10) edge node{} (a47);

\path (a11) edge node{} (a56);
\path (a11) edge node{} (a57);

 \end{tikzpicture}		
 \end{center}		
	\caption{The binary tree $t_3$ from Lemma~\ref{lemma-vergleich2} (left) and its first-child next-sibling encoding $\ffcns(t_3)$ (right)}	
	\label{fig-t3}
\end{figure}

\begin{proof}
Let $\Sigma=\{a,b,c,d,e\}$.
We define the tree  $t_n$ as $t_n = a(b(d^n)c(d(e)^n))$. That is, $t_n$ is a tree consisting of a root node of degree two, whose left child is of degree $n$ and labeled $b$ and whose right child is of degree $n$ and  labeled $c$. Moreover, the children of the left child of the root are leaves labeled with $d$ and the children of the right child of the root are unary nodes labeled with $d$, whose children are leaves labeled with $e$. 
The tree $t_3$ is depicted in Figure~\ref{fig-t3} on the left. We have $|t_n|=3n+3$.
Let $\Box = a$ denote the fixed dummy symbol used to pad histories of length shorter than $k$. 
We start with computing $H_k^{\ell}(t_n)$. There are three nodes $v$ with $k$-label-history $\lh_k(v)=a^k$: The root node (labeled with $a$) and its two children (one of them labeled with $b$, one of them labeled with $c$). Moreover, there are $n$ nodes with $k$-label-history $a^{k-1}b$ (all of them labeled with $d$) and $n$ nodes with $k$-label-history $a^{k-1}c$ (all of them labeled with $d$). Finally, there are $n$ nodes with $k$-label-history $a^{k-2}cd$, all of them are labeled with $e$. We obtain:
\begin{equation*}
H_k^{\ell}(t_n) = \sum_{z \in \Sigma^k}\sum_{a \in \Sigma}\tildem_{z,a}^t\log_2\left(\frac{\tildem_z^t}{\tildem_{z,a}^t}\right) = 3\log_2(3).
\end{equation*}
In order to compute the label-degree history of $t_n$, we observe that the $k$-label-history and the label of a node of $t_n$ uniquely determine the degree of the node.
This implies $H_k^{\ell,\deg}(t_n)=0$. Altogether, this yields
\begin{align*}
H_k^{\ell}(t_n)+H_k^{\ell,\deg}(t_n)=3\log_2(3).
\end{align*}
Next, we compute $H^{\deg}(t_n)$: A tree $t_n$ consists of a node of degree $2$ (the root node), two nodes of degree $n$ (the two children of the root node), $n$ unary nodes and $2n$ leaf nodes. Thus, the degree entropy satisfies
\begin{eqnarray*}
  H^{\deg}(t_n) &=& \sum_{i=1}^{|t|}n_i^t\log_2\left(\frac{|t|}{n_i^t}\right)\\
&=& \log_2(3n+3)+2\log_2\left(\frac{3n+3}{2}\right)+ \\
& & n\log_2\left(\frac{3n+3}{n}\right)+ 2n\log_2\left(\frac{3n+3}{2n}\right) \\
&\geq & 2n.
\end{eqnarray*}
We thus have $H^{\deg}(t_n) + H_k^{\deg,\ell}(t_n) \geq 2n$.
It remains to lower-bound $H_k(t_n)$. For this, we have to consider the first-child next-sibling encoding of $t_n$. Let $\Box=a$ denote the dummy symbol used to label the leaf nodes in $\ffcns(t_n)$.
The tree $\ffcns(t_3)$ is depicted in Figure~\ref{fig-t3} on the right.
We lower-bound $H_k(t_n)$ by considering only nodes in $\ffcns(t_n)$ with $k$-history $(d1)^{k-1}d0$.
There are $2(n-k+1)$ such nodes and half of them are labeled with $a$ while the other half of them are labeled with $e$.
We thus have $H_k(t_n) \geq 2(n-k+1)$.
\qed
\end{proof}
Note that we clearly need $\Omega(\log n)$ bits to represent the tree $t_n$  from the above proof 
(since we have to represent its size). 
This does not contradict Theorem~\ref{theorem-ganczorzentropybounds} and 
the $\mathcal{O}(1)$-bound for $H_k^{\ell}(t_n)+H_k^{\ell,\deg}(t_n)$ in  Lemma~\ref{lemma-vergleich2}, since we have the additional
additive term of order $o(|t|)$ in Theorem~\ref{theorem-ganczorzentropybounds}.

In the following lemma, $n^{\underline{k}} = n (n-1) \cdots (n-k+1)$ denotes the falling factorial.
\begin{lemma} \label{lemma-vergleich-last}
There exists a family of unranked trees $(t_{n,k})_{n \geq 1}$ (where $k(n) \leq n$ may depend on $n$) 
such that for all $n \geq 1$:
\begin{enumerate}[(i)]
\item $|t_{n,k}| = 1 + n^{\underline{k}} + k \cdot n \cdot n^{\underline{k}}$,
\item $H^{\deg}(t_{n,k}) + H_1^{\ell}(t_{n,k}) \leq  \mathcal{O}(n \cdot n^{\underline{k}}\cdot k \cdot \log k )$ and
\item  $H_{k-1}(t_{n,k}) \geq  \Omega(  n \cdot n^{\underline{k}} \cdot k \cdot \log(n-k+1) )$.
\end{enumerate}
\end{lemma}

\begin{figure}[t]
\begin{center}
		\tikzset{level 1/.style={sibling distance=22mm}}
		\tikzset{level 2/.style={sibling distance=3.5mm}}
		\begin{tikzpicture}[scale=.9,auto,swap,level distance=8mm]
		\node (r) {\scriptsize $a$} 
		child {node {\scriptsize $b_{1,2}$}
		   child {node {\scriptsize $c_{1}$}}
		   child {node {\scriptsize $c_{2}$}}   
		   child {node {\scriptsize $c_{1}$}}
		   child {node {\scriptsize $c_{2}$}}
		   child {node {\scriptsize $c_{1}$}}
		   child {node {\scriptsize $c_{2}$}}}
	        child {node {\scriptsize $b_{2,1}$}
		   child {node {\scriptsize $c_{2}$}}
		   child {node {\scriptsize $c_{1}$}}   
		   child {node {\scriptsize $c_{2}$}}
		   child {node {\scriptsize $c_{1}$}}
		   child {node {\scriptsize $c_{2}$}}
		   child {node {\scriptsize $c_{1}$}}}
		child {node {\scriptsize $b_{1,3}$}
		   child {node {\scriptsize $c_{1}$}}
		   child {node {\scriptsize $c_{3}$}}   
		   child {node {\scriptsize $c_{1}$}}
		   child {node {\scriptsize $c_{3}$}}
		   child {node {\scriptsize $c_{1}$}}
		   child {node {\scriptsize $c_{3}$}}}
	        child {node {\scriptsize $b_{3,1}$}
		   child {node {\scriptsize $c_{3}$}}
		   child {node {\scriptsize $c_{1}$}}   
		   child {node {\scriptsize $c_{3}$}}
		   child {node {\scriptsize $c_{1}$}}
		   child {node {\scriptsize $c_{3}$}}
		   child {node {\scriptsize $c_{1}$}}}
	        child {node {\scriptsize $b_{2,3}$}
		   child {node {\scriptsize $c_{2}$}}
		   child {node {\scriptsize $c_{3}$}}   
		   child {node {\scriptsize $c_{2}$}}
		   child {node {\scriptsize $c_{3}$}}
		   child {node {\scriptsize $c_{2}$}}
		   child {node {\scriptsize $c_{3}$}}}
	        child {node {\scriptsize $b_{3,2}$}
		   child {node {\scriptsize $c_{3}$}}
		   child {node {\scriptsize $c_{2}$}}   
		   child {node {\scriptsize $c_{3}$}}
		   child {node {\scriptsize $c_{2}$}}
		   child {node {\scriptsize $c_{3}$}}
		   child {node {\scriptsize $c_{2}$}}};	   
 \end{tikzpicture}		
 \end{center}		
	\caption{The tree $t_{3,2}$ from Lemma~\ref{lemma-vergleich-last}.}	
	\label{fig-t32}
\end{figure}
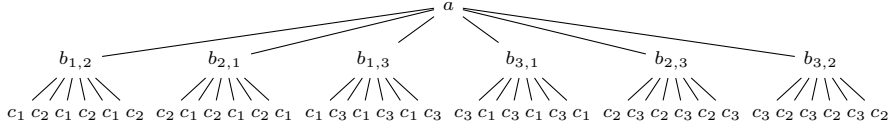

\begin{proof}
Let $[n]^{\underline{k}} = \{ (i_1,i_2\ldots,i_k) \mid 1 \leq i_1, \ldots, i_k \leq n, i_j \neq i_l \text{ for } j \neq l \}$.
The tree $t_{n,k}$ is defined over the label alphabet 
$$
\Sigma_{n,k} = \{a\} \cup \{ b_{u} \mid u \in [n]^{\underline{k}} \} \cup \{ c_i \mid 1 \leq i \leq n\}   .
$$
For $u = (i_1, i_2, \ldots, i_k) \in [n]^{\underline{k}}$ let us define the tree
$t_u = b_u( (c_{i_1} c_{i_2} \cdots c_{i_k})^n)$ and let
$$
t_{n,k} = a(t_{u_1} t_{u_2} \cdots t_{u_m}),
$$ 
where $u_1, u_2, \ldots, u_m$ is an arbitrary enumeration of the set $[n]^{\underline{k}}$ (hence, $m = n^{\underline{k}}$).
The tree $t_{3,2}$ is shown in Figure~\ref{fig-t32}.
We have $|t_{n,k}| = 1 + n^{\underline{k}} + k \cdot n \cdot n^{\underline{k}}$. 

Let us first compute $H^{\deg}(t_{n,k})$. There are (i) one node of degree $n^{\underline{k}}$, (ii) $n^{\underline{k}}$ nodes
of degree $k \cdot n$ and  (iii) $k \cdot n \cdot n^{\underline{k}}$ leaves. Hence, we obtain
\begin{eqnarray*}
H^{\deg}(t_{n,k})  & = & \log_2(1 + n^{\underline{k}} + k \cdot n \cdot n^{\underline{k}}) + n^{\underline{k}} \cdot \log_2 \left( \frac{1 + n^{\underline{k}} + k \cdot n \cdot n^{\underline{k}}}{n^{\underline{k}} }\right)
+ \\
&& k \cdot n \cdot n^{\underline{k}} \cdot \log_2 \left( \frac{1 + n^{\underline{k}} + k \cdot n \cdot n^{\underline{k}}}{k \cdot n \cdot n^{\underline{k}}}\right) \\
& = & \log_2(k \cdot n \cdot n^{\underline{k}}) + n^{\underline{k}} \cdot \log_2\left(1 + k \cdot n + \frac{1}{n^{\underline{k}}}\right) + \\
& & k \cdot n \cdot n^{\underline{k}} \cdot  \log_2\left(1 + \frac{1}{k \cdot n} +  \frac{1}{k \cdot n \cdot n^{\underline{k}}}\right) + \mathcal{O}(1) \\
& = &  \log_2(k \cdot n \cdot n^{\underline{k}}) +  n^{\underline{k}} \cdot \log_2(k \cdot n) + k \cdot n \cdot n^{\underline{k}} \cdot  \log_2\left(1 +  \frac{1}{k \cdot n}\right) + \mathcal{O}(1)  \\
& \leq &  \log_2(k \cdot n \cdot n^{\underline{k}}) +  n^{\underline{k}} \cdot \log_2(k \cdot n) +  \mathcal{O}(n^{\underline{k}})  \\
& \leq & \mathcal{O}(n^{\underline{k}} \cdot \log n), 
\end{eqnarray*}
where we used $k \leq n$ and the inequality $\log_2(1+x) \leq x/\ln 2$ for $x \geq 0$.

Next, we compute $H_1^{\ell}(t_{n,k})$. As usual, we choose $\Box = a$ for the padding symbol.
There are $1+n^{\underline{k}}$ nodes with $1$-label-history $a$ (the root and its children), which are labeled with pairwise 
different symbols. Moreover, for every $u = (i_1, i_2, \ldots, i_k) \in [n]^{\underline{k}}$
there are $kn$ nodes with $1$-label-history $a_u$, of which $n$ nodes are labeled with $a_{i_j}$ for every $1 \leq j \leq k$.
We therefore obtain
\begin{eqnarray*}
H^{\ell}_1(t_{n,k})  & = & (1+n^{\underline{k}}) \cdot \log_2(1+n^{\underline{k}}) + \sum_{u \in [n]^{\underline{k}}} \sum_{i=1}^k n \cdot \log_2 k \\
& = & (1+n^{\underline{k}}) \cdot \log_2(1+n^{\underline{k}}) + n \cdot n^{\underline{k}} \cdot k \cdot \log_2 k \\
& \leq &  \mathcal{O}( n \cdot n^{\underline{k}} \cdot k \cdot \log_2 k )
\end{eqnarray*}
In summary, we get $H^{\deg}(t_{n,k})+ H^{\ell}_1(t_{n,k}) \leq \mathcal{O}( n \cdot n^{\underline{k}}\cdot k \cdot \log k)$.

It remains to bound $H_{k-1}(t_{n,k})$. For this we only consider $(k-1)$-histories of the form
$c_{i_1} 1 c_{i_2} 1 \cdots c_{i_{k-1}}1$, where $(i_1,i_2\ldots,i_{k-1}) \in [n]^{\underline{k-1}}$.
Let us denote this $(k-1)$-history by $z_u$, where $u = (i_1,i_2\ldots,i_{k-1})$. In the following, we identify the nodes of $t_{n,k}$
with the inner nodes of $\ffcns(t_{n,k})$.
For every symbol $c_j$ such that $j \notin \{ i_1, \ldots, i_{k-1}\}$ there are $n + (n-1) (k-1)$ nodes
$v$ with $(k-1)$-history $z_u$ in $\ffcns(t_{n,k})$ and $\lambda(v) = (c_j,2)$: $n$ children of $b_{u'}$ where $u' = (i_1, \ldots, i_{k-1},j)$ and 
$n-1$ children of $b_{u''}$ where $u''$ is a cyclic rotation of $u'$ with $u'' \neq u'$ (there are $(k-1)$
such rotations). Moreover, there are $n-k+1$ nodes with  $(k-1)$-history $z_u$ and $\lambda(v) = (a,0)$.
In total we obtain $(n + (n-1) (k-1)) (n-k+1) + (n-k+1) = (n-k+1) (n + 1 + (n-1) (k-1))$ nodes with $(k-1)$-history $z_u$.
By computing the contribution of these nodes (for all $u \in [n]^{\underline{k-1}}$) to $H_{k-1}(t_{n,k})$ we obtain
\begin{eqnarray*}
& & H_{k-1}(t_{n,k}) \\ 
& \geq & \sum_{u \in [n]^{\underline{k-1}}} (n-k+1) (n + (n-1) (k-1))  \log_2 \left( n-k+1 + \frac{n-k+1}{n + (n-1) (k-1)} \right) + \\
& &  \sum_{u \in [n]^{\underline{k-1}}} (n-k+1) \log_2( n + 1 + (n-1) (k-1)) \\
& \geq & n^{\underline{k}} \cdot ((n + (n-1) (k-1))  \log_2(n-k+1) +  \log_2( n + 1 + (n-1) (k-1))) \\
& \geq & \Omega(  n \cdot n^{\underline{k}}  \cdot k \cdot \log(n-k+1) ).
\end{eqnarray*}
This concludes the proof of the lemma.
\qed
\end{proof}
If $k \in (\log n)^{\mathcal{O}(1)}$ then the trees $t_{n,k}$ from Lemma~\ref{lemma-vergleich-last} satisfy
\begin{equation} \label{eq-separation}
\frac{H^{\deg}(t_{n,k}) + H_1^{\ell}(t_{n,k})}{H_k(t_{n,k})} \leq \mathcal{O}\left( \frac{\log k}{\log(n-k+1)} \right)  = o(1).
\end{equation}

\section{Experiments}

\begin{table}[h!]
\centering\scriptsize
    \begin{tabular}{l || R{0.5cm} | R{2cm} | R{2cm} | R{2cm} | R{2cm}}
    XML          & $k$ & $H_k$ & $H^{\deg}+H^{\ell}_k$ & $H^{\ell}_k+H^{\ell,\deg}_k$ & $H^{\deg}+H^{\deg,\ell}_k$ \\
    \midrule
    \midrule
\fname{BaseBall}&0&202\ 568.08&153\ 814.94&146\ 066.64&146\ 066.64\\
&1&6\ 348.08&145\ 705.73&137\ 957.42&145\ 323.26\\
&2&2\ 671.95&145\ 705.73&137\ 957.42&145\ 323.26\\
&4&1\ 435.11&145\ 705.73&137\ 957.42&145\ 323.26\\
\midrule
\fname{DBLP}&0&18\ 727\ 523.44&14\ 576\ 781.00&12\ 967\ 501.16&12\ 967\ 501.16\\
&1&2\ 607\ 784.68&12\ 137\ 042.56&10\ 527\ 690.38&12\ 076\ 935.39\\
&2&2\ 076\ 410.50&12\ 136\ 974.71&10\ 527\ 595.96&12\ 076\ 845.69\\
&4&1\ 951\ 141.63&12\ 136\ 966.29&10\ 527\ 586.31&12\ 076\ 836.82\\
\midrule
\fname{EXI-Array}&0&1\ 098\ 274.54&962\ 858.05&649\ 410.59&649\ 410.59\\
&1&4\ 286.39&387\ 329.51&73\ 882.05&387\ 304.76\\
&2&4\ 270.18&387\ 329.51&73\ 882.05&387\ 304.76\\
&4&4\ 263.82&387\ 329.51&73\ 882.05&387\ 304.76\\
\midrule
\fname{EXI-factbook}&0&530\ 170.92&481\ 410.05&423\ 012.12&423\ 012.12\\
&1&11\ 772.65&239\ 499.01&181\ 101.08&204\ 649.84\\
&2&5\ 049.98&239\ 499.01&181\ 101.08&204\ 649.84\\
&4&4\ 345.42&239\ 499.01&181\ 101.08&204\ 649.84\\
\midrule
\fname{EnWikiNew}&0&2\ 118\ 359.59&1\ 877\ 639.22&1\ 384\ 034.65&1\ 384\ 034.65\\
&1&243\ 835.84&1\ 326\ 743.94&833\ 139.36&1\ 095\ 837.20\\
&2&78\ 689.86&1\ 326\ 743.94&833\ 139.36&1\ 095\ 837.20\\
&4&78\ 687.51&1\ 326\ 743.94&833\ 139.36&1\ 095\ 837.20\\
\midrule
\fname{EnWikiQuote}&0&1\ 372\ 201.38&1\ 229\ 530.04&894\ 768.55&894\ 768.55\\
&1&156\ 710.30&871\ 127.39&536\ 365.91&717\ 721.09\\
&2&51\ 557.50&871\ 127.39&536\ 365.91&717\ 721.09\\
&4&51\ 557.31&871\ 127.39&536\ 365.91&717\ 721.09\\
\midrule
\fname{EnWikiVersity}&0&2\ 568\ 158.43&2\ 264\ 856.93&1\ 644\ 997.36&1\ 644\ 997.36\\
&1&278\ 832.56&1\ 594\ 969.93&975\ 110.35&1\ 311\ 929.24\\
&2&74\ 456.55&1\ 594\ 969.93&975\ 110.35&1\ 311\ 929.24\\
&4&74\ 456.41&1\ 594\ 969.93&975\ 110.35&1\ 311\ 929.24\\
\midrule
\fname{Nasa}&0&3\ 022\ 100.11&2\ 872\ 172.41&2\ 214\ 641.55&2\ 214\ 641.55\\
&1&292\ 671.36&1\ 368\ 899.76&701\ 433.91&1\ 226\ 592.72\\
&2&168\ 551.10&1\ 363\ 699.16&696\ 194.53&1\ 221\ 474.16\\
&4&147\ 041.08&1\ 363\ 699.16&696\ 194.53&1\ 221\ 474.16\\
\midrule
\fname{Shakespeare}&0&655\ 517.90&521\ 889.47&395\ 890.85&395\ 890.85\\
&1&138\ 283.88&370\ 231.89&244\ 047.64&347\ 212.36\\
&2&125\ 837.77&370\ 061.20&243\ 843.87&347\ 041.31\\
&4&123\ 460.80&370\ 057.77&243\ 838.09&347\ 037.86\\
\midrule
\fname{SwissProt}&0&18\ 845\ 126.39&16\ 063\ 648.44&13\ 755\ 427.39&13\ 755\ 427.39\\
&1&3\ 051\ 570.48&11\ 065\ 924.67&8\ 757\ 703.61&10\ 238\ 734.83\\
&2&2\ 634\ 911.88&11\ 065\ 924.67&8\ 757\ 703.61&10\ 238\ 734.83\\
&4&2\ 314\ 609.48&11\ 065\ 924.67&8\ 757\ 703.61&10\ 238\ 734.83\\
\midrule
\fname{Treebank}&0&16\ 127\ 202.92&15\ 669\ 672.80&12\ 938\ 625.09&12\ 938\ 625.09\\
&1&7\ 504\ 481.18&12\ 301\ 414.61&9\ 482\ 695.67&9\ 925\ 567.44\\
&2&5\ 607\ 499.40&11\ 909\ 330.06&9\ 051\ 186.33&9\ 559\ 968.40\\
&4&4\ 675\ 093.61&11\ 626\ 935.89&8\ 736\ 301.14&9\ 285\ 544.85\\
\midrule
\fname{USHouse}&0&36\ 266.08&34\ 369.06&28\ 381.43&28\ 381.43\\
&1&10\ 490.44&24\ 249.78&17\ 968.41&19\ 438.19\\
&2&9\ 079.97&24\ 037.34&17\ 569.59&19\ 216.99\\
&4&6\ 308.98&23\ 634.87&16\ 830.00&18\ 783.36\\
\midrule
\fname{XMark1}&0&1\ 250\ 525.41&1\ 186\ 214.34&988\ 678.93&988\ 678.93\\
&1&167\ 586.81&592\ 634.17&394\ 639.43&523\ 996.29\\
&2&131\ 057.35&592\ 625.76&394\ 565.79&523\ 969.97\\
&4&127\ 157.34&592\ 037.39&393\ 770.73&523\ 432.87
\end{tabular}
  \caption{A comparison of the upper bounds on the bits used by the data structures in~\cite{HuckeLS19} (third column) and~\cite{Ganczorz20} (columns 4, 5 and 6) where lower order terms are ignored.} \label{table1}
\end{table}

In this section we complement our theoretical results with experimental data.
We computed the entropies $H^{\deg}$, $H_k$,  $H^{\ell}_k$, $H^{\ell,\deg}_k$ and $H^{\deg,\ell}_k$ (for $k \in \{0,1,2,4\}$) for 13 XML files from XMLCompBench (\url{http://xmlcompbench.sourceforge.net}).
In Table~\ref{table1} we compare the upper bounds (ignoring lower order terms) on the bits needed by the compressed data structures from~\cite{HuckeLS19} ($H_k$; see also Theorem~\ref{theorem-isitpaper}) and~\cite{Ganczorz20} ($H^{\deg}+H^{\ell}_k$, $H^{\ell}_k+H^{\ell,\deg}_k$ and $H^{\deg}+H^{\deg,\ell}_k$; see also Theorem~\ref{theorem-ganczorzentropybounds}).
It turns out that for all XML trees used in this comparison the $k^{th}$-order label-shape entropy (for $k>0$) from~\cite{HuckeLS19} is significantly smaller than the 
entropies from~\cite{Ganczorz20}.

In Appendix~\ref{appendix-exp} (Table~\ref{table2}) the reader finds all tree entropy measures discussed in this paper for each XML (divided by the tree size so that the table fits on the page).
Additionally, we computed the label-shape entropy $H_k$ for a modified version of each XML where all labels are replaced by a single dummy symbol, i.e., we considered the underlying, unlabeled tree as well (in Table~\ref{table2} this value is denoted by $H'_k$).
Note again that the label-shape entropy is the only measure where this modification is interesting because (i) the degree entropy $H^{\deg}$ is not affected since it does not take labels into account and
(ii) we have $H^{\ell,\deg}_k(t)=H^{\deg}(t)$ and $H^{\ell}_k(t)=H^{\deg,\ell}_k(t)=0$ for all unlabeled trees $t$  and for all $k$.
In the setting of unlabeled trees, our experimental data indicates that neither the label-shape entropy nor the degree entropy (which is the upper bound on the number of bits needed by the data structure in~\cite{JanssonSS12} ignoring lower order terms; see also Theorem~\ref{theorem-degentropybound}) is favorable.

\section{Open problems}

The separation between $H_k$ and $H_1^\ell + H^{\deg}$ 
achieved in Lemma~\ref{lemma-vergleich-last} is quite weak: for a constant $k$, $H_k$ is only by a logarithmic factor larger than 
$H_1^\ell + H^{\deg}$; see \eqref{eq-separation}. In contrast, in Lemmas~\ref{lemma-vergleich1} and~\ref{lemma-vergleich2} we achieved an exponential separation. It remains
open, whether such an exponential separation is also possible for $H_k$ and $H_1^\ell + H^{\deg}$. In other words,
does there exist a family of trees $t_n$ such that $H_k(t_n) \in \Omega(n)$ and $H_1^\ell(t_n) + H^{\deg}(t_n) \in \mathcal{O}(\log n)$?

\bibliographystyle{plain}
\bibliography{bib}

\newpage
\appendix
\section{Additional experimental data} \label{appendix-exp}
The following table shows the entropy measures discussed in this work for real XML data.\footnote{We want to remark that in~\cite{Ganczorz20} the values $H^{\deg}/n$, $H^{\ell}_k/n$, $H^{\deg,\ell}_k/n$ and $H^{\ell,\deg}_k/n$ are computed for \fname{EnWikiNew.xml}, \fname{Nasa.xml} and \fname{Treebank.xml}, but for the documents \fname{EnWikiNew.xml} and \fname{Nasa.xml} the presented values are incorrect.}
\begin{table}[h!]
\centering\tiny
    \begin{tabular}{l || R{1.2cm} | R{1.2cm} | R{0.5cm} | R{1.2cm} | R{1.2cm} | R{1.2cm} | R{1.2cm} | R{1.2cm}}
    XML             & $n$  & $H^{\deg}/n$&$k$& $H'_k/n$&$H_k/n$ & $H^{\ell}_k/n$ & $H^{\deg,\ell}_k/n$ & $H^{\ell,\deg}_k/n$   \\
    \midrule
    \midrule
    \fname{BaseBall}&28\ 306&0.2777&0&2.0000&7.1564&5.1563&4.8826&0.0039\\
&&&1&0.5271&0.2243&4.8698&4.8563&0.0039\\
&&&2&0.5218&0.0944&4.8698&4.8563&0.0039\\
&&&4&0.5122&0.0507&4.8698&4.8563&0.0039\\
\midrule
\fname{DBLP}&3\ 332\ 130&0.7543&0&2.0000&5.6203&3.6203&3.1373&0.2714\\
&&&1&0.9343&0.7826&2.8881&2.8701&0.2713\\
&&&2&0.9064&0.6231&2.8881&2.8700&0.2713\\
&&&4&0.8340&0.5856&2.8881&2.8700&0.2713\\
\midrule
\fname{EXI-Array}&226\ 523&1.4022&0&2.0000&4.8484&2.8484&1.4647&0.0185\\
&&&1&1.9736&0.0189&0.3077&0.3076&0.0185\\
&&&2&1.8227&0.0189&0.3077&0.3076&0.0185\\
&&&4&0.4486&0.0188&0.3077&0.3076&0.0185\\
\midrule
\fname{EXI-factbook}&55\ 453&1.1207&0&2.0000&9.5607&7.5607&6.5076&0.0676\\
&&&1&1.2641&0.2123&3.1983&2.5698&0.0676\\
&&&2&1.2319&0.0911&3.1983&2.5698&0.0676\\
&&&4&1.1811&0.0784&3.1983&2.5698&0.0676\\
\midrule
\fname{EnWikiNew}&404\ 652&1.4051&0&2.0000&5.2350&3.2350&2.0152&0.1853\\
&&&1&1.6514&0.6026&1.8736&1.3030&0.1853\\
&&&2&1.3977&0.1945&1.8736&1.3030&0.1853\\
&&&4&1.0771&0.1945&1.8736&1.3030&0.1853\\
\midrule
\fname{EnWikiQuote}&262\ 955&1.4574&0&2.0000&5.2184&3.2184&1.9453&0.1844\\
&&&1&1.6878&0.5960&1.8554&1.2720&0.1844\\
&&&2&1.4695&0.1961&1.8554&1.2720&0.1844\\
&&&4&1.0229&0.1961&1.8554&1.2720&0.1844\\
\midrule
\fname{EnWikiVersity}&495\ 839&1.3883&0&2.0000&5.1794&3.1794&1.9293&0.1382\\
&&&1&1.6647&0.5623&1.8284&1.2576&0.1382\\
&&&2&1.4106&0.1502&1.8284&1.2576&0.1382\\
&&&4&0.9645&0.1502&1.8284&1.2576&0.1382\\
\midrule
\fname{Nasa}&476\ 646&1.6855&0&2.0000&6.3403&4.3403&2.9608&0.3060\\
&&&1&1.8834&0.6140&1.1865&0.8879&0.2851\\
&&&2&1.8483&0.3536&1.1756&0.8772&0.2850\\
&&&4&1.3824&0.3085&1.1756&0.8772&0.2850\\
\midrule
\fname{Shakespeare}&179\ 690&1.2563&0&2.0000&3.6480&1.6480&0.9468&0.5551\\
&&&1&1.3713&0.7696&0.8040&0.6759&0.5541\\
&&&2&1.2713&0.7003&0.8031&0.6750&0.5539\\
&&&4&1.1215&0.6871&0.8031&0.6750&0.5539\\
\midrule
\fname{SwissProt}&2\ 977\ 031&1.0657&0&2.0000&6.3302&4.3302&3.5548&0.2903\\
&&&1&1.2108&1.0250&2.6514&2.3736&0.2903\\
&&&2&1.0730&0.8851&2.6514&2.3736&0.2903\\
&&&4&1.0553&0.7775&2.6514&2.3736&0.2903\\
\midrule
\fname{Treebank}&2\ 437\ 666&1.8123&0&2.0000&6.6158&4.6158&3.4955&0.6920\\
&&&1&1.9707&3.0786&3.2341&2.2594&0.6560\\
&&&2&1.8014&2.3004&3.0732&2.1095&0.6398\\
&&&4&1.7620&1.9179&2.9574&1.9969&0.6265\\
\midrule
\fname{USHouse}&6\ 712&1.7175&0&2.0000&5.4032&3.4030&2.5109&0.8254\\
&&&1&1.8263&1.5629&1.8954&1.1785&0.7817\\
&&&2&1.5810&1.3528&1.8637&1.1456&0.7539\\
&&&4&1.2958&0.9400&1.8038&1.0810&0.7037\\
\midrule
\fname{XMark1}&167\ 865&1.6169&0&2.0000&7.4496&5.4496&4.2728&0.4401\\
&&&1&1.6917&0.9983&1.9135&1.5046&0.4374\\
&&&2&1.6820&0.7807&1.9135&1.5045&0.4370\\
&&&4&1.5735&0.7575&1.9100&1.5013&0.4358
\end{tabular}
  \caption{Experimental results for XML tree structures, where $n$ denotes the number of nodes and $H'_k$ is the label-shape entropy for the underlying, unlabeled tree.}  \label{table2}
\end{table}

\end{document}